\documentclass[11pt]{article}

\usepackage[a4paper,margin=30mm]{geometry}
\usepackage[round,authoryear]{natbib}
\usepackage{microtype}
\usepackage{pp-manuscript}
\usepackage[hidelinks]{hyperref}
\usepackage[capitalise,nameinlink,noabbrev]{cleveref}
\AddToHook{file/primary_profile_support.tex/before}{\clearpage}

\newcommand{\papertitle}{Information Games: Strategic Crowding and Firm
Repositioning in Language-Model Space}

\newcommand{\paperauthors}{Marcus Gawronsky, Chun-Sung Huang}

\newcommand{\paperaffiliations}{Department of Finance and Tax,
  University of Cape
Town}

\newcommand{\paperdate}{October 2026}

\newcommand{\paperkeywords}{Competitive adjustment; Rival feedback; Operating
exposure; Distributional measurement}
\newcommand{\paperjel}{D43; G14; L25; C63}

\hypersetup{pdftitle={\papertitle}}
\ppDeclareValue{ir-finite-affine}{2.7281961570176431e-05}
\ppDeclareValue{ir-solve-target}{1e-10}
\ppDeclareValue{ir-failures}{0}
\ppDeclareValue{ir-maximum-solver-bound}{9.9993218048603457e-11}
\ppDeclareValue{ir-minimum-capacity}{0.0092340308645421452}
\ppDeclareValue{ir-minimum-baseline-price}{0.050491338684610421}
\ppDeclareValue{ir-active-domain-violations}{0}
\ppDeclareValue{ir-common-1-bias}{0.094776499388490754}
\ppDeclareValue{ir-common-1-bias-se}{0.24775335874264187}
\ppDeclareValue{ir-common-1-rmse}{16.889460895725115}
\ppDeclareValue{ir-common-1-rmse-se}{0.28500500195702444}
\ppDeclareValue{ir-common-1-undefined}{0}
\ppDeclareValue{ir-common-4-bias}{-0.082793227450754694}
\ppDeclareValue{ir-common-4-bias-se}{0.13111287901668545}
\ppDeclareValue{ir-common-4-rmse}{8.948482759480644}
\ppDeclareValue{ir-common-4-rmse-se}{0.14552629023076624}
\ppDeclareValue{ir-common-4-undefined}{0}
\ppDeclareValue{ir-common-16-bias}{-0.010522107517604695}
\ppDeclareValue{ir-common-16-bias-se}{0.067065573003628476}
\ppDeclareValue{ir-common-16-rmse}{4.5845481863231985}
\ppDeclareValue{ir-common-16-rmse-se}{0.075919628565484967}
\ppDeclareValue{ir-common-16-undefined}{0}
\ppDeclareValue{ir-changing-1-bias}{0.24766133981245042}
\ppDeclareValue{ir-changing-1-bias-se}{0.25360460306293714}
\ppDeclareValue{ir-changing-1-rmse}{17.302181181312317}
\ppDeclareValue{ir-changing-1-rmse-se}{0.30790458569950147}
\ppDeclareValue{ir-changing-1-undefined}{0}
\ppDeclareValue{ir-changing-4-bias}{0.73899431918248681}
\ppDeclareValue{ir-changing-4-bias-se}{0.13369696782402185}
\ppDeclareValue{ir-changing-4-rmse}{9.2077627343576314}
\ppDeclareValue{ir-changing-4-rmse-se}{0.15282030840769112}
\ppDeclareValue{ir-changing-4-undefined}{0}
\ppDeclareValue{ir-changing-16-bias}{0.66869571134870154}
\ppDeclareValue{ir-changing-16-bias-se}{0.068772476929650675}
\ppDeclareValue{ir-changing-16-rmse}{4.8200687048154967}
\ppDeclareValue{ir-changing-16-rmse-se}{0.080215426858603639}
\ppDeclareValue{ir-changing-16-undefined}{0}
\ppDeclareValue{ir-dimension-2-1-rmse}{18.123270251241909}
\ppDeclareValue{ir-dimension-2-1-se}{0.30866791453730719}
\ppDeclareValue{ir-dimension-2-16-rmse}{8.9900297269205058}
\ppDeclareValue{ir-dimension-2-16-se}{0.15296200091918344}
\ppDeclareValue{ir-dimension-4-1-rmse}{31.510133942587125}
\ppDeclareValue{ir-dimension-4-1-se}{0.56713910845412485}
\ppDeclareValue{ir-dimension-4-16-rmse}{11.441816995427468}
\ppDeclareValue{ir-dimension-4-16-se}{0.19905550832408395}
\ppDeclareValue{ir-dimension-8-1-rmse}{95.676025207530401}
\ppDeclareValue{ir-dimension-8-1-se}{1.6057875060427762}
\ppDeclareValue{ir-dimension-8-16-rmse}{30.553635832155823}
\ppDeclareValue{ir-dimension-8-16-se}{0.53328560020451721}
\ppDeclareValue{ir-dimension-16-1-rmse}{131.68181094232654}
\ppDeclareValue{ir-dimension-16-1-se}{2.4187725967782026}
\ppDeclareValue{ir-dimension-16-16-rmse}{38.107710918168216}
\ppDeclareValue{ir-dimension-16-16-se}{0.66064162845636709}
\ppDeclareValue{ir-reporting-only-positive-histories}{16}
\ppDeclareValue{ir-common-1-positive-pct}{52.685546875}
\ppDeclareValue{ir-common-1-positive-se-pct}{1.1044379757565521}
\ppDeclareValue{ir-common-1-sign-undefined}{0}
\ppDeclareValue{ir-common-4-positive-pct}{55.322265625}
\ppDeclareValue{ir-common-4-positive-se-pct}{1.0938656791900658}
\ppDeclareValue{ir-common-4-sign-undefined}{0}
\ppDeclareValue{ir-common-16-positive-pct}{59.47265625}
\ppDeclareValue{ir-common-16-positive-se-pct}{1.0683306550816984}
\ppDeclareValue{ir-common-16-sign-undefined}{0}
\ppDeclareValue{ir-changing-1-positive-pct}{52.5390625}
\ppDeclareValue{ir-changing-1-positive-se-pct}{1.1009223156941244}
\ppDeclareValue{ir-changing-1-sign-undefined}{0}
\ppDeclareValue{ir-changing-4-positive-pct}{59.47265625}
\ppDeclareValue{ir-changing-4-positive-se-pct}{1.0701741648184542}
\ppDeclareValue{ir-changing-4-sign-undefined}{0}
\ppDeclareValue{ir-changing-16-positive-pct}{68.603515625}
\ppDeclareValue{ir-changing-16-positive-se-pct}{0.99024047195625708}
\ppDeclareValue{ir-changing-16-sign-undefined}{0}
\ppDeclareValue{ir-validation-histories}{16}
\ppDeclareValue{ir-validation-redraws}{160}
\ppDeclareValue{ir-validation-common-1-correct-pct}{53.0859375}
\ppDeclareValue{ir-validation-common-1-positive-pct}{53.0859375}
\ppDeclareValue{ir-validation-common-1-positive-se-pct}{0.98678805234686562}
\ppDeclareValue{ir-validation-common-1-positive-denominator}{2560}
\ppDeclareValue{ir-validation-common-1-rmse}{11.269559609434811}
\ppDeclareValue{ir-validation-common-1-mean-target}{0.60022980702219875}
\ppDeclareValue{ir-validation-common-1-normalized-rmse}{17.003354742844792}
\ppDeclareValue{ir-validation-common-4-correct-pct}{55.546874999999993}
\ppDeclareValue{ir-validation-common-4-positive-pct}{55.546874999999993}
\ppDeclareValue{ir-validation-common-4-positive-se-pct}{0.97767723327016032}
\ppDeclareValue{ir-validation-common-4-rmse}{6.0127198700272606}
\ppDeclareValue{ir-validation-common-4-mean-target}{0.60022980702219875}
\ppDeclareValue{ir-validation-common-4-normalized-rmse}{9.0719080835983501}
\ppDeclareValue{ir-validation-common-16-correct-pct}{59.843749999999993}
\ppDeclareValue{ir-validation-common-16-positive-pct}{59.843749999999993}
\ppDeclareValue{ir-validation-common-16-positive-se-pct}{0.95310365413255727}
\ppDeclareValue{ir-validation-common-16-rmse}{3.0352587480546047}
\ppDeclareValue{ir-validation-common-16-mean-target}{0.60022980702219875}
\ppDeclareValue{ir-validation-common-16-normalized-rmse}{4.5795561688398321}
\ppDeclareValue{ir-validation-changing-16-correct-pct}{68.125}
\ppDeclareValue{ir-validation-stationary-1-positive-pct}{49.6875}
\ppDeclareValue{ir-validation-stationary-1-positive-se-pct}{0.98726933295829733}
\ppDeclareValue{ir-validation-stationary-1-rmse}{11.639921096739526}
\ppDeclareValue{ir-validation-stationary-1-mean-target}{0}
\ppDeclareValue{ir-validation-reporting-only-positive-pct}{52.617187499999993}
\ppDeclareValue{ir-validation-reporting-only-positive-se-pct}{0.98852733243661661}
\ppDeclareValue{ir-validation-reporting-only-rmse}{11.386767683607664}
\ppDeclareValue{ir-validation-reporting-only-mean-target}{0.68648889175149885}
\ppDeclareValue{ir-validation-half-correct-pct}{51.640624999999993}
\ppDeclareValue{ir-validation-double-correct-pct}{54.765624999999993}
\ppDeclareValue{ir-inference-z}{1.6448536269514715}
\ppDeclareValue{ir-validation-0-label}{Mechanism changing 1}
\ppDeclareValue{ir-validation-0-target}{Operating}
\ppDeclareValue{ir-validation-0-correct-pct}{52.9296875}
\ppDeclareValue{ir-validation-0-correct-se-pct}{0.98290979698067382}
\ppDeclareValue{ir-validation-0-positive-pct}{52.9296875}
\ppDeclareValue{ir-validation-0-positive-se-pct}{0.98290979698067382}
\ppDeclareValue{ir-validation-0-rmse}{11.529304771999326}
\ppDeclareValue{ir-validation-1-label}{Mechanism changing 4}
\ppDeclareValue{ir-validation-1-target}{Operating}
\ppDeclareValue{ir-validation-1-correct-pct}{59.492187500000007}
\ppDeclareValue{ir-validation-1-correct-se-pct}{0.95636543671375107}
\ppDeclareValue{ir-validation-1-positive-pct}{59.492187500000007}
\ppDeclareValue{ir-validation-1-positive-se-pct}{0.95636543671375107}
\ppDeclareValue{ir-validation-1-rmse}{6.0848532510321869}
\ppDeclareValue{ir-validation-2-label}{Mechanism changing 16}
\ppDeclareValue{ir-validation-2-target}{Operating}
\ppDeclareValue{ir-validation-2-correct-pct}{68.125}
\ppDeclareValue{ir-validation-2-correct-se-pct}{0.8893711693091636}
\ppDeclareValue{ir-validation-2-positive-pct}{68.125}
\ppDeclareValue{ir-validation-2-positive-se-pct}{0.8893711693091636}
\ppDeclareValue{ir-validation-2-rmse}{3.1714378940576453}
\ppDeclareValue{ir-validation-3-label}{Mechanism common 1}
\ppDeclareValue{ir-validation-3-target}{Operating}
\ppDeclareValue{ir-validation-3-correct-pct}{53.0859375}
\ppDeclareValue{ir-validation-3-correct-se-pct}{0.98678805234686562}
\ppDeclareValue{ir-validation-3-positive-pct}{53.0859375}
\ppDeclareValue{ir-validation-3-positive-se-pct}{0.98678805234686562}
\ppDeclareValue{ir-validation-3-rmse}{11.269559609434811}
\ppDeclareValue{ir-validation-4-label}{Mechanism common 4}
\ppDeclareValue{ir-validation-4-target}{Operating}
\ppDeclareValue{ir-validation-4-correct-pct}{55.546874999999993}
\ppDeclareValue{ir-validation-4-correct-se-pct}{0.97767723327016032}
\ppDeclareValue{ir-validation-4-positive-pct}{55.546874999999993}
\ppDeclareValue{ir-validation-4-positive-se-pct}{0.97767723327016032}
\ppDeclareValue{ir-validation-4-rmse}{6.0127198700272606}
\ppDeclareValue{ir-validation-5-label}{Mechanism common 16}
\ppDeclareValue{ir-validation-5-target}{Operating}
\ppDeclareValue{ir-validation-5-correct-pct}{59.843749999999993}
\ppDeclareValue{ir-validation-5-correct-se-pct}{0.95310365413255727}
\ppDeclareValue{ir-validation-5-positive-pct}{59.843749999999993}
\ppDeclareValue{ir-validation-5-positive-se-pct}{0.95310365413255727}
\ppDeclareValue{ir-validation-5-rmse}{3.0352587480546047}
\ppDeclareValue{ir-validation-6-label}{Stationary law common 1}
\ppDeclareValue{ir-validation-6-target}{Reported}
\ppDeclareValue{ir-validation-6-positive-pct}{49.6875}
\ppDeclareValue{ir-validation-6-positive-se-pct}{0.98726933295829733}
\ppDeclareValue{ir-validation-6-rmse}{11.639921096739526}
\ppDeclareValue{ir-validation-7-label}{Stationary law common 4}
\ppDeclareValue{ir-validation-7-target}{Reported}
\ppDeclareValue{ir-validation-7-positive-pct}{51.249999999999993}
\ppDeclareValue{ir-validation-7-positive-se-pct}{0.98835742607642951}
\ppDeclareValue{ir-validation-7-rmse}{6.0819691479533935}
\ppDeclareValue{ir-validation-8-label}{Stationary law common 16}
\ppDeclareValue{ir-validation-8-target}{Reported}
\ppDeclareValue{ir-validation-8-positive-pct}{49.8046875}
\ppDeclareValue{ir-validation-8-positive-se-pct}{0.98795438847039729}
\ppDeclareValue{ir-validation-8-rmse}{3.0983229419214302}
\ppDeclareValue{ir-validation-9-label}{Reporting only}
\ppDeclareValue{ir-validation-9-target}{Reported}
\ppDeclareValue{ir-validation-9-correct-pct}{52.617187499999993}
\ppDeclareValue{ir-validation-9-correct-se-pct}{0.98852733243661661}
\ppDeclareValue{ir-validation-9-positive-pct}{52.617187499999993}
\ppDeclareValue{ir-validation-9-positive-se-pct}{0.98852733243661661}
\ppDeclareValue{ir-validation-9-rmse}{11.386767683607664}
\ppDeclareValue{ir-validation-10-label}{Reporting only}
\ppDeclareValue{ir-validation-10-target}{Operating}
\ppDeclareValue{ir-validation-10-positive-pct}{52.617187499999993}
\ppDeclareValue{ir-validation-10-positive-se-pct}{0.98852733243661661}
\ppDeclareValue{ir-validation-10-rmse}{11.401122346050318}
\ppDeclareValue{ir-validation-11-label}{Double shock}
\ppDeclareValue{ir-validation-11-target}{Operating}
\ppDeclareValue{ir-validation-11-correct-pct}{54.765624999999993}
\ppDeclareValue{ir-validation-11-correct-se-pct}{0.97822676583451096}
\ppDeclareValue{ir-validation-11-positive-pct}{54.765624999999993}
\ppDeclareValue{ir-validation-11-positive-se-pct}{0.97822676583451096}
\ppDeclareValue{ir-validation-11-rmse}{5.8088245761884929}
\ppDeclareValue{ir-validation-12-label}{Half shock}
\ppDeclareValue{ir-validation-12-target}{Operating}
\ppDeclareValue{ir-validation-12-correct-pct}{51.640624999999993}
\ppDeclareValue{ir-validation-12-correct-se-pct}{0.98515773262785011}
\ppDeclareValue{ir-validation-12-positive-pct}{51.640624999999993}
\ppDeclareValue{ir-validation-12-positive-se-pct}{0.98515773262785011}
\ppDeclareValue{ir-validation-12-rmse}{22.583382116294139}

\ppDeclareValue{ir-firms}{100}
\ppDeclareValue{ir-coordinates}{39600}
\ppDeclareValue{ir-redraws}{199}
\ppDeclareValue{ir-native-dimension}{2560}
\ppDeclareValue{ir-expanded-articles}{43720}
\ppDeclareValue{ir-matched-articles}{12800}
\ppDeclareValue{ir-primary-beta}{0.0047686432783675741}
\ppDeclareValue{ir-financial-example}{0.01}
\ppDeclareValue{ir-primary-percent-effect}{4.7686432783675742e-05}
\ppDeclareValue{ir-primary-low}{0.0017519782567947508}
\ppDeclareValue{ir-primary-high}{0.0072364981583044148}
\ppDeclareValue{ir-year-deletion-low}{0.0029267121784888901}
\ppDeclareValue{ir-year-deletion-high}{0.0057961039308053307}
\ppDeclareValue{ir-role-deletion-low}{0.0037169822197166201}
\ppDeclareValue{ir-role-deletion-high}{0.005695051231250952}
\ppDeclareValue{ir-reference-share-pct}{15.103534794151532}
\ppDeclareValue{ir-target-share-pct}{1.2349443043130577}
\ppDeclareValue{ir-primary-2020-weight-pct}{41.382949055702738}
\ppDeclareValue{ir-normalized-alignment}{0.036281947789676944}
\ppDeclareValue{ir-annual-2019-beta}{0.002959740861695326}
\ppDeclareValue{ir-annual-2020-beta}{0.0073776542739866962}
\ppDeclareValue{ir-annual-2021-beta}{0.0047659855168913617}
\ppDeclareValue{ir-annual-2022-beta}{0.0019451605449275198}
\ppDeclareValue{ir-d2560-expanded-beta}{0.0047686432783675741}
\ppDeclareValue{ir-d2560-expanded-low}{0.0017519782567947508}
\ppDeclareValue{ir-d2560-expanded-high}{0.0072364981583044148}
\ppDeclareValue{ir-d2560-expanded-q05}{0.002195023397462336}
\ppDeclareValue{ir-d2560-expanded-alignment}{0.036281947789676944}
\ppDeclareValue{ir-d2560-expanded-within-fit}{0.0013163797354128436}
\ppDeclareValue{ir-d2560-matched-beta}{0.0055522262282938217}
\ppDeclareValue{ir-d2560-matched-low}{0.00098747959866315798}
\ppDeclareValue{ir-d2560-matched-high}{0.0090794177310423874}
\ppDeclareValue{ir-d2560-matched-q05}{0.0018384199299945351}
\ppDeclareValue{ir-d2560-matched-alignment}{0.028271931489184859}
\ppDeclareValue{ir-d64-expanded-beta}{0.011319143868042641}
\ppDeclareValue{ir-d64-expanded-low}{0.0092376372308312395}
\ppDeclareValue{ir-d64-expanded-high}{0.01348057321956737}
\ppDeclareValue{ir-d64-expanded-q05}{0.0097350510462708342}
\ppDeclareValue{ir-d64-expanded-alignment}{0.084534588047318462}
\ppDeclareValue{ir-d64-matched-beta}{0.011975116963378142}
\ppDeclareValue{ir-d64-matched-low}{0.008271722558457064}
\ppDeclareValue{ir-d64-matched-high}{0.015765842984807304}
\ppDeclareValue{ir-d64-matched-q05}{0.009025281938351646}
\ppDeclareValue{ir-d64-matched-alignment}{0.05908445875387499}

\ppDeclareValue{ir-observation-primary-heldout-pct}{1.4736250074561696}
\ppDeclareValue{ir-observation-locality-count}{663}
\ppDeclareValue{ir-observation-locality-defined}{784}
\ppDeclareValue{ir-observation-locality-pct}{84.566326530612244}
\ppDeclareValue{ir-observation-gmm-beta}{0.0054041364311125722}
\ppDeclareValue{ir-observation-gmm-2020-weight-pct}{58.194401599266953}
\ppDeclareValue{ir-observation-equal-budget-pct}{1.7979345463195997}
\ppDeclareValue{ir-observation-heldout-2019}{0.0019376996855776651}
\ppDeclareValue{ir-observation-heldout-2020}{-0.004839040645791659}
\ppDeclareValue{ir-observation-heldout-2021}{-0.0098575279890031897}
\ppDeclareValue{ir-observation-heldout-2022}{0.058806466857451833}
\ppDeclareValue{ir-observation-d2560-expanded-heldout-pct}{1.4736250074561696}
\ppDeclareValue{ir-observation-d2560-matched-heldout-pct}{-0.95555246779046588}
\ppDeclareValue{ir-observation-d64-expanded-heldout-pct}{0.76257518620820308}
\ppDeclareValue{ir-observation-d64-matched-heldout-pct}{-2.0793333423943183}

\ppDeclareValue{ir-redraw-histories}{16}
\ppDeclareValue{ir-redraw-outer-replications}{160}
\ppDeclareValue{ir-redraw-inner-redraws}{199}
\ppDeclareValue{ir-redraw-signal-first-response-min-pct}{0.031571363773115306}
\ppDeclareValue{ir-redraw-signal-first-response-max-pct}{0.050227558934379246}
\ppDeclareValue{ir-redraw-signal-maximum-origin-change-pct}{0.20343851335852259}
\ppDeclareValue{ir-redraw-mechanism-central-coverage-pct}{96.796875}
\ppDeclareValue{ir-redraw-mechanism-central-coverage-se-pct}{0.34742996111349395}
\ppDeclareValue{ir-redraw-mechanism-central-coverage-attempts}{2560}
\ppDeclareValue{ir-redraw-mechanism-central-coverage-available}{2560}
\ppDeclareValue{ir-redraw-mechanism-q05-coverage-pct}{95.78125}
\ppDeclareValue{ir-redraw-mechanism-q05-coverage-se-pct}{0.39717870968530011}
\ppDeclareValue{ir-redraw-mechanism-positive-q05-pct}{4.765625}
\ppDeclareValue{ir-redraw-mechanism-positive-q05-se-pct}{0.42110359971881961}
\ppDeclareValue{ir-redraw-stationary-central-coverage-pct}{95.468749999999986}
\ppDeclareValue{ir-redraw-stationary-central-coverage-se-pct}{0.41191025357091482}
\ppDeclareValue{ir-redraw-stationary-q05-coverage-pct}{95.625}
\ppDeclareValue{ir-redraw-stationary-q05-coverage-se-pct}{0.40395875094447781}
\ppDeclareValue{ir-redraw-stationary-positive-q05-pct}{4.375}
\ppDeclareValue{ir-redraw-stationary-positive-q05-se-pct}{0.40395875094447781}
\ppDeclareValue{ir-redraw-reporting-central-coverage-pct}{95.742187500000014}
\ppDeclareValue{ir-redraw-reporting-central-coverage-se-pct}{0.3995516009170792}
\ppDeclareValue{ir-redraw-reporting-q05-coverage-pct}{95.2734375}
\ppDeclareValue{ir-redraw-reporting-q05-coverage-se-pct}{0.4193108110164338}
\ppDeclareValue{ir-redraw-reporting-positive-q05-pct}{5.859375}
\ppDeclareValue{ir-redraw-reporting-positive-q05-se-pct}{0.46362419544815914}

\ppDeclareValue{ir-inference-audit-n16-firms}{16}
\ppDeclareValue{ir-inference-audit-n16-null-target}{0}
\ppDeclareValue{ir-inference-audit-n16-null-coverage-pct}{89.480000000000004}
\ppDeclareValue{ir-inference-audit-n16-null-coverage-se-pct}{0.61374466433949715}
\ppDeclareValue{ir-inference-audit-n16-null-rejection-pct}{10.52}
\ppDeclareValue{ir-inference-audit-n16-null-rejection-se-pct}{0.61374466433949726}
\ppDeclareValue{ir-inference-audit-n16-null-se-ratio}{0.80059716124014169}
\ppDeclareValue{ir-inference-audit-n16-null-attempts}{2500}
\ppDeclareValue{ir-inference-audit-n16-null-available}{2500}
\ppDeclareValue{ir-inference-audit-n16-alternative-target}{0.5}
\ppDeclareValue{ir-inference-audit-n16-alternative-coverage-pct}{89.480000000000004}
\ppDeclareValue{ir-inference-audit-n16-alternative-coverage-se-pct}{0.61374466433949715}
\ppDeclareValue{ir-inference-audit-n16-alternative-rejection-pct}{69.359999999999999}
\ppDeclareValue{ir-inference-audit-n16-alternative-rejection-se-pct}{0.92218020285116298}
\ppDeclareValue{ir-inference-audit-n16-alternative-se-ratio}{0.80059716124014169}
\ppDeclareValue{ir-inference-audit-n16-coverage-pct}{89.480000000000004}
\ppDeclareValue{ir-inference-n100-firms}{100}
\ppDeclareValue{ir-inference-n100-years}{4}
\ppDeclareValue{ir-inference-n100-replications}{2500}
\ppDeclareValue{ir-inference-n100-null-target}{0}
\ppDeclareValue{ir-inference-n100-null-coverage-pct}{93.799999999999997}
\ppDeclareValue{ir-inference-n100-null-coverage-se-pct}{0.48240759408899653}
\ppDeclareValue{ir-inference-n100-null-rejection-pct}{6.2000000000000002}
\ppDeclareValue{ir-inference-n100-null-rejection-se-pct}{0.48240759408899636}
\ppDeclareValue{ir-inference-n100-null-se-ratio}{0.95514991318670839}
\ppDeclareValue{ir-inference-n100-alternative-target}{0.5}
\ppDeclareValue{ir-inference-n100-alternative-coverage-pct}{93.799999999999997}
\ppDeclareValue{ir-inference-n100-alternative-coverage-se-pct}{0.48240759408899653}
\ppDeclareValue{ir-inference-n100-alternative-rejection-pct}{99.839999999999989}
\ppDeclareValue{ir-inference-n100-alternative-rejection-se-pct}{0.079951966372134281}
\ppDeclareValue{ir-inference-n100-alternative-se-ratio}{0.95514991318670861}
\ppDeclareValue{ir-inference-n100-coverage-pct}{93.799999999999997}

\ppDeclareValue{lean-verified-prop:equilibrium}{checked}
\ppDeclareValue{lean-disclosure-prop:equilibrium}{The capacity equilibrium statement is checked for nonnegative rational-share primitives and symmetric uniformly positive cost curvature; the Gaussian benchmark economy is a specialization. Symmetry and the independent counterfactual are supported by the separate entries below.}
\ppDeclareValue{lean-verified-app:independent-proof}{checked}
\ppDeclareValue{lean-disclosure-app:independent-proof}{The independent arm is checked as a product of one-firm problems, each facing outside appeal alone.}
\ppDeclareValue{lean-verified-app:symmetry-proof}{checked}
\ppDeclareValue{lean-disclosure-app:symmetry-proof}{The symmetry conclusion concerns firms within one economy, not equality between strategic and independent economies.}
\ppDeclareValue{lean-verified-eq:certificate}{checked}
\ppDeclareValue{lean-disclosure-eq:certificate}{The joint residual supplies a distance certificate for the strategic capacity equilibrium.}
\ppDeclareValue{lean-verified-app:regret-proof}{checked}
\ppDeclareValue{lean-disclosure-app:regret-proof}{The residual bounds payoff regret, not only an optimizer stopping criterion.}
\ppDeclareValue{lean-verified-eq:normalization-bound}{checked}
\ppDeclareValue{lean-disclosure-eq:normalization-bound}{The bound composes joint-to-block error control with normalization and ensures positive exact mass.}
\ppDeclareValue{lean-verified-app:independent-distance}{checked}
\ppDeclareValue{lean-disclosure-app:independent-distance}{The certificate applies to the Cartesian product of no-rival objectives.}
\ppDeclareValue{lean-verified-app:independent-regret}{checked}
\ppDeclareValue{lean-disclosure-app:independent-regret}{The certificate applies to the Cartesian product of no-rival objectives.}
\ppDeclareValue{lean-verified-eq:response-continuity}{checked}
\ppDeclareValue{lean-disclosure-eq:response-continuity}{The Lipschitz bound is checked for any two equilibria that differ only in opportunity values, with no differentiability, interiority, or active-set condition.}
\ppDeclareValue{lean-verified-app:response-jacobian}{checked}
\ppDeclareValue{lean-disclosure-app:response-jacobian}{The Jacobian blocks are checked as the derivative of the field along every line through a profile with positive denominators.}
\ppDeclareValue{lean-verified-prop:response}{checked}
\ppDeclareValue{lean-disclosure-prop:response}{The response is checked as the unique solution of the equilibrium conditions linearized on any fixed active set; its identification with the derivative of the equilibrium map is checked separately under strict complementarity.}
\ppDeclareValue{lean-verified-eq:linear-response}{checked}
\ppDeclareValue{lean-disclosure-eq:linear-response}{The bound uses curvature of the Jacobian at every nonnegative profile, not only at an equilibrium.}
\ppDeclareValue{lean-verified-eq:feedback}{checked}
\ppDeclareValue{lean-disclosure-eq:feedback}{The identity is checked in solved-system form; the explicit inverse follows from uniqueness of the linearized response.}
\ppDeclareValue{lean-verified-eq:composition-derivative}{checked}
\ppDeclareValue{lean-disclosure-eq:composition-derivative}{The normalization derivative is checked coordinatewise at any block with positive mass.}
\ppDeclareValue{lean-verified-cor:exposure}{checked}
\ppDeclareValue{lean-disclosure-cor:exposure}{The exposure identity is exact at fixed capacities; it does not include the equilibrium response of capacities.}
\ppDeclareValue{lean-verified-eq:exposure-covariance}{checked}
\ppDeclareValue{lean-disclosure-eq:exposure-covariance}{The covariance formula is checked for square-integrable innovations passed through a predetermined share matrix.}
\ppDeclareValue{lean-verified-eq:rival-cross}{checked}
\ppDeclareValue{lean-disclosure-eq:rival-cross}{The nonpositive rival cross-effect is checked in finite-difference form.}
\ppDeclareValue{lean-verified-app:anchored-stage}{checked}
\ppDeclareValue{lean-disclosure-app:anchored-stage}{Existence and uniqueness of the anchored stage equilibrium are checked with no sign restriction on the anchor-shifted linear cost.}
\ppDeclareValue{lean-verified-eq:equilibrium-derivative}{checked}
\ppDeclareValue{lean-disclosure-eq:equilibrium-derivative}{The first-order approximation of the equilibrium map by the linear response is checked under strict complementarity, without an implicit-function theorem.}
\ppDeclareValue{lean-verified-app:equilibrium-derivative}{checked}
\ppDeclareValue{lean-disclosure-app:equilibrium-derivative}{Fréchet differentiability of the field and the identification of its derivative with the Jacobian blocks are checked.}
\ppDeclareValue{lean-verified-eq:envelope}{checked}
\ppDeclareValue{lean-disclosure-eq:envelope}{The three-term derivative of cash flow along any joint direction is checked.}
\ppDeclareValue{lean-verified-rem:envelope}{checked}
\ppDeclareValue{lean-disclosure-rem:envelope}{The static envelope split is checked at an equilibrium where the own-adjustment term vanishes.}
\ppDeclareValue{lean-verified-eq:equilibrium-cash-flow}{checked}
\ppDeclareValue{lean-disclosure-eq:equilibrium-cash-flow}{The equilibrium cash-flow sensitivity is checked as a first-order approximation under strict complementarity; the firm's own adjustment contributes nothing at first order.}
\ppDeclareValue{lean-verified-eq:cost-response}{checked}
\ppDeclareValue{lean-disclosure-eq:cost-response}{The first-order approximation of the equilibrium map in linear costs is checked under strict complementarity, without an implicit-function theorem.}
\ppDeclareValue{lean-verified-eq:two-site-transport}{checked}
\ppDeclareValue{lean-disclosure-eq:two-site-transport}{Machine-checked by a lower bound for every nonnegative finite coupling and construction of an attaining coupling. The result concerns the finite two-site transport problem.}
\ppDeclareValue{lean-verified-eq:proportional-feedback}{checked}
\ppDeclareValue{lean-disclosure-eq:proportional-feedback}{Exact normalization equivalence. Identification with equilibrium derivatives separately requires the manuscript's differentiable fixed-active-set path premises; no implicit-function theorem is claimed here.}
\ppDeclareValue{lean-verified-prop:composition-feedback}{checked}
\ppDeclareValue{lean-disclosure-prop:composition-feedback}{Supporting sign result for the explicit rank-one response system. The equilibrium-path identification, activity separability and unaffected-activity interpretation remain explicit human-proof premises. This binding does not enlarge coverage to arbitrary multi-activity or active-set-changing equilibria.}

\title{\papertitle}
\author{\paperauthors\\\paperaffiliations}
\date{\paperdate}

\begin{document}

\maketitle

\begin{abstract}
  Firms reallocate capacity as market opportunities change, altering competitors'
profitability and the cash flows investors anticipate.
We extend the multi-contest framework of
\citet{xu_zenou_zhou_conflict_2022} to explain how competition and adjustment
costs shape firms' changing positions and their exposure to free cash flow
disturbances.
The model establishes equilibrium properties and comparative statics, linking
gradual operating adjustment to financial responses when new economic
opportunities arrive.

In a restricted two-opportunity economy, firms with greater initial exposure
experience larger return responses to positive exogenous free cash flow shocks.
As other firms expand toward that opportunity, their news distributions
approach the initial distributions of those competitors.
Under common reporting intensity and the stated timing conditions, this
mechanism implies a nonnegative expected local association between financial
returns and movement toward competitors in information space.

We develop a novel sandwich estimator that relates changes in Wasserstein
distances between
corporate-news distributions to excess returns weighted by article timing.
Multi-agent simulations show how limited coverage obscures small operating
adjustments and how changing coverage can generate similar observed movements.
In an application using Nasdaq-listed firms our empirical study finds a
small positive association across alternative specifications.
Together, these results give an economic foundation for studying firms'
changing economic footprints through financial news and clarifies how competition
and information production shape firms' investment decisions.

\end{abstract}


\medskip

\noindent\textit{JEL classification:} \paperjel

\section{Introduction}\label{sec:introduction}

A manufacturer responding to stronger demand for industrial automation must
decide how much capacity to add and which activities to expand.
Its investment changes the margins available to competing manufacturers,
whose responses in turn affect the attractiveness of that expansion.
Adjusting capacity takes time, while news of its cash-flow consequences can
change firms' valuations before the operating response is complete.
These linked decisions change firms' positions relative to their competitors
and create a connection between investment, financial valuation, and the
information reported about firms.

Firms engage in several activities at once, and their competitive positions
depend on the composition of those activities as well as their scale.
This allocation of capacity defines a firm's economic footprint.
Recent work by \citet{gawronsky_continuous_2026} represents firms through
distributions of financial-news embeddings, preserving differences among the
articles associated with each firm rather than reducing its coverage to a
single average.
Such distributions provide a way to compare firms' positions in information space.
Understanding how firms' economic footprints evolve requires an account of
the choices that change their activities.

Spatial and multi-market competition provide a foundation for that account.
Since \citet{hotelling_stability_1929}, models of product positioning have
linked firms' locations to the demand they serve.
Research on strategic substitutes and complements and on networked
competition shows how firms' choices affect rivals across markets
\citep{bulow_geanakoplos_klemperer_multimarket_1985,bimpikis_ehsani_ilkilic_cournot_2019}.
The multi-contest framework of \citet{xu_zenou_zhou_conflict_2022} formalizes
such interdependence when players allocate resources across overlapping contests.
We adapt that framework to firms allocating operating capacity across
activities that serve related customer demands, where each firm's expansion
changes the contribution margins available to its competitors.

Firms differ in the activities they can pursue at low cost, while increasing
capacity brings additional expense.
We establish conditions for a unique Nash equilibrium and derive how
opportunity and cost changes affect capacity allocations, business mix, and
rivals' responses.
To study repositioning over time, we introduce costs of changing inherited capacity.
Each period's operating choices form a Nash equilibrium conditional on the
previous allocation, so adjustment unfolds gradually as competitors respond
to one another.
Investors value the cash flows generated along this path at a constant discount factor.
They incorporate anticipated changes in cash flows into firms'
valuations immediately, even as operating capacity adjusts over subsequent
periods.

In a restricted two-opportunity economy, firms with greater initial exposure
experience larger return responses to positive exogenous free cash flow shocks.
The shock improves one opportunity relative to the other, inducing firms to
reallocate capacity toward it.
If news coverage reflects the relative allocation of operating capacity, a
firm increasing its exposure becomes closer to the initial news distributions
of more-exposed competitors and farther from those of less-exposed competitors.
With fixed reference distributions and the stated conditions on news timing
and observation errors, these responses imply a nonnegative expected local
association between relative financial returns and decreases in the target
firm's distance to other firms.
This result gives competitive repositioning an observable implication in
information space.

Financial news is itself produced through economic choices.
Investors seek information about material changes in investment, products,
and cash-flow prospects, while publishers and corporate disclosures influence
which developments receive coverage \citep{gentzkow_media_bias_2006}.
Following research that uses corporate text to study product-market
relationships \citep{hoberg_phillips_2016}, we represent a firm's articles as
a mixture of distributions associated with its activities.
The relative weight of each activity depends on both operating capacity and
the intensity of coverage.
When coverage per unit of capacity is equal across activities, news
composition follows the firm's business mix.
Changes in coverage can also alter the reported distribution while capacity
remains unchanged.

We measure differences between firms' article distributions using squared
Wasserstein distance.
Optimal transport finds a minimum-cost matching between the distributions of
article embeddings, capturing changes in their composition even when their
average vectors change little.
For each firm, we measure whether its annual news distribution moves closer
to or farther from other firms' distributions fixed at their 2018 positions.
We compare these changes with the reference firms' daily excess returns,
summed with weights determined by the proportion of the target firm's
articles published after each return date.
Centering both measures across reference firms within a target-year produces
an estimator of the predicted relative association.

Multi-agent simulations assess this estimator using operating paths and
return responses generated by the model.
Under equal reporting intensity, the population coefficient is positive in
each of sixteen simulated news histories, yet only
\ppnum[1]{ir-validation-common-1-correct-pct}\% of article samples at
observed coverage yield a numerically resolved positive estimate.
Increasing article coverage improves recovery, although measurement error
remains substantial.
In a separate observation counterfactual, changing the relative coverage of
activities also produces positive alignment while operating capacities are held fixed.
These experiments show how article sampling and the production of news affect
the ability to observe competitive adjustment.

We apply the estimator to financial news and returns for \ppnum[0]{ir-firms}
Nasdaq-listed firms from 2018 to 2022.
Across four specifications, firms' movements relative to fixed reference
distributions are positively associated with those reference firms' relative
weighted excess returns, although the relationship is small.
The evidence is descriptive, while the simulations establish why changes in
news cannot be interpreted as changes in operating capacity alone.
Together, the theory and empirical analysis extend the study of
distribution-valued firms from comparisons of their reported economic
footprints to an economic account of why those footprints may change over time.

The remainder of the paper proceeds as follows.
Section~\ref{sec:literature} reviews the related literature.
Section~\ref{sec:model} develops the economic model, its directional
prediction, the estimator and the restricted-form hypothesis.
Section~\ref{sec:application} evaluates the estimator and its measurement
uncertainty through multi-agent simulations, then describes the corporate data
and empirical design.
Section~\ref{sec:results} presents the results and their economic interpretation.
Section~\ref{sec:conclusion} concludes.

\section{Literature Review}\label{sec:literature}

Spatial competition provides a foundation for understanding how firms'
positions influence demand and the profitability of their decisions.
In \citet{hotelling_stability_1929}, consumers choose between firms by
weighing prices against transportation costs, making location an important
determinant of competitive advantage.
A firm's decision to reposition consequently depends on the customers it can
attract and the responses of its competitors.
These interactions become more complex when firms operate across several
markets. \citet{bulow_geanakoplos_klemperer_multimarket_1985} establish how
strategic substitutes and complements connect competitive decisions, while
\citet{bimpikis_ehsani_ilkilic_cournot_2019} examine how competition and
production linkages transmit shocks through networked markets.
Together, these contributions motivate studying firm positioning through the
allocation of resources across related activities and the competitive
responses those allocations generate.

The multi-contest framework of \citet{xu_zenou_zhou_conflict_2022} provides
the closest theoretical foundation for our analysis.
In their model, players allocate costly effort across contests, receiving
shares of prizes that depend on their allocations relative to those of competitors.
Their equilibrium characterization and comparative statics establish how
competitive interactions transmit changes in incentives through a network of contests.
We adapt this proportional-share structure to firms allocating operating
capacity across activities that serve overlapping customer demands.
Capacity determines a firm's ability to serve each demand opportunity, while
rivals compete for the same available contribution margins.
The resulting allocation problem admits a discrete-choice interpretation when
capacity represents varieties or outlets
\citep{mcfadden_conditional_logit_1974,train_discrete_choice_2009}.
Building on the equilibrium and sensitivity methods of
\citet{xu_zenou_zhou_conflict_2022}, we establish uniqueness and characterize
firms' responses to changes in market opportunities and costs under the
conditions of our capacity model.

Competitive repositioning also depends on the costs of changing existing
investments. \citet{sweeting_positioning_2013} studies forward-looking
product positioning in differentiated markets, demonstrating the importance
of adjustment costs for firms' decisions.
We introduce persistence through quadratic costs of changing inherited
capacity, with firms making sequential Nash choices as market opportunities
and competitors' allocations evolve.
The resulting operating path determines the cash flows investors anticipate
following an exogenous free cash flow shock. \citet{berk_green_naik_growth_options_1999}
relate investment opportunities to security returns through changes in assets
and growth options.
We follow this connection between operating opportunities and financial
valuation, discounting the cash flows generated along the adjustment path.
Firm valuations can consequently respond when new information arrives, while
firms' operating allocations adjust over time.

Corporate disclosures and financial news provide information about firms'
activities and their competitive environments. \citet{hoberg_phillips_2016}
demonstrate how corporate text can reveal product-market relationships and
differences in competitive positioning, while
\citet{gentzkow_text_as_data_2019} establish methodological foundations for
using text as economic data.
The information available to investors also reflects the incentives governing
its production. \citet{gentzkow_media_bias_2006} study how reputation and
reporting incentives influence news content, providing an economic motivation
for considering differences in the coverage of firms' activities.
We incorporate these considerations through an observation model in which a
firm's news distribution depends on both its operating capacity and the
relative intensity of reporting across activities.

Representing corporate information as a distribution preserves variation
among the activities and perspectives reflected in individual articles.
\citet{kusner_word_embeddings_2015} develop transport-based comparisons of
word embeddings, while \citet{gawronsky_continuous_2026} use Wasserstein
geometry to study covariance envelopes from language-model representations.
Optimal transport compares distributions through the minimum cost of matching
their probability mass, retaining information about their composition and
dispersion that may be lost when articles are represented by a single average vector.
This provides a natural basis for comparing firms' positions in information
space and measuring changes in those positions over time.

We extend this literature by deriving an economic mechanism for changes in
firms' economic footprints.
Costly capacity adjustment and competitive responses determine how firms
react to changing opportunities, while financial valuation reflects the
cash-flow consequences of those decisions.
In a restricted two-opportunity economy, initial exposure determines the
ordering of firms' return responses and, under the specified reporting
conditions, the direction of movement toward other firms' initial news distributions.
The resulting nonnegative expected local association provides the theoretical
basis for the estimator developed in the following section.

\section{Methodology}\label{sec:model}

\subsection{Competitive Environment}\label{sec:model-payoff}

We adapt the multi-contest framework of
\citet{xu_zenou_zhou_conflict_2022} to firms allocating capacity
across activities serving overlapping customer demands.
Capacity generates appeal at each demand node, and a firm's appeal relative
to rivals and outside alternatives determines its share of the available
contribution margin.
In the static game, each firm chooses capacity to maximize that margin net
of capacity costs, taking rivals' allocations as given.

Let firm $i=1,\ldots,N$ assign capacity $q_{ia}\geq0$ to activity
$a=1,\ldots,m$, located at $x_a\in\mathbb R^d$.
Customer demand arises at $L$ nodes $z_\ell\in\mathbb R^d$.
These locations describe economic activities and demands, while article-vector
coordinates will represent their reported profiles in
Section~\ref{sec:adjustment-design}.
To distinguish scale from business mix, we represent the allocation by a
measure $Q_i$ with total mass $S_i$ and normalized composition $P_i$:
\begin{equation}\label{eq:position}
  Q_i=\sum_a q_{ia}\delta_{x_a},\qquad
  S_i=\sum_{a}q_{ia},\qquad
  P_i=\sum_a p_{ia}\delta_{x_a},\quad p_{ia}=q_{ia}/S_i,
\end{equation}
where $\delta_{x_a}$ is a unit point mass at $x_a$ and $P_i$ is defined when
$S_i>0$.
The vector $p_i$ records the composition of simultaneous activities, not
probabilities of selecting a single activity.
Keeping total capacity $S_i$ separate allows the same disturbance to change
operating scale and business mix differently.
Capacities are in common model units, and the joint allocation
$q=(q_1,\ldots,q_N)$ lies in $\mathcal Q=\mathbb R_+^{Nm}$
(Appendix~\ref{app:proofs}).

Capacity in one activity can serve several related customer segments.
The coefficient $K_{\ell a}=\exp\{-\tau_A\|x_a-z_\ell\|^2/d\}$ measures how
effectively activity $a$ serves demand node $\ell$, declining with the
distance between them, and a higher $\tau_A$ makes it decline faster and narrows
the range of demands a given activity serves.
A firm's appeal at a node is the sum of its activity capacities weighted by
that effectiveness, $A_i=Kq_i\in\mathbb R_+^L$, and each firm faces the
combined appeal of its rivals and of an outside alternative, measured in the
same units:
\begin{equation}\label{eq:appeal}
  A_{i\ell}(q_i)=\sum_{a}K_{\ell a}q_{ia},\qquad
  B_{i\ell}(q_{-i})=b_\ell+\sum_{j\ne i}A_{j\ell}(q_j),\qquad b_\ell>0.
\end{equation}
The outside appeal $b_\ell$ represents alternatives beyond the modeled firms
and keeps the share denominator positive even when no modeled firm
serves a node.
Own capacity increases a firm's access to a demand opportunity, while rivals
divide that opportunity among more participants.

At each demand node, firms compete for a contribution-margin opportunity
$r_{\ell t}\geq0$ available at date $t$, valued in cash-flow units.
Exogenous free cash flow shocks change these available contribution margins.
Firms' initial exposures and subsequent capacity choices determine the
resulting changes in their cash flows.
A firm receives the fraction of that opportunity given by its share of total
appeal, including the outside alternative, so its operating payoff is
contribution margin less the persistent cost
$C_i$ of maintaining capacity:
\begin{equation}\label{eq:firm-problem}
  \pi_i(q_i,q_{-i};r_t)=
  \underbrace{\sum_\ell r_{\ell t}
  \frac{A_{i\ell}(q_i)}{A_{i\ell}(q_i)+B_{i\ell}(q_{-i})}}
  _{\text{contribution margin from current opportunities}}
  -\underbrace{C_i(q_i)}_{\text{cost of positioning there}}.
\end{equation}
Each firm chooses nonnegative capacity to maximize this stylized unlevered
operating cash flow, taking its rivals' allocations as given.
Its contribution-margin component, denoted $\mathcal M_i$
(Appendix~\ref{app:payoff}), adapts the proportional-share contest of
\citet{xu_zenou_zhou_conflict_2022}: capacities generate appeal through $K$,
and $b_\ell$ represents alternatives outside the modeled firms.
Appendix~\ref{app:contest-inheritance} gives the adaptation and its required
assumptions.
Doubling an opportunity weight doubles its contribution to the margin at an
unchanged profile, whereas doubling every firm's capacity need not double the
margin because shares approach saturation.
Summed over firms, the margins at a node never exceed $r_{\ell t}$: firms
divide a finite pool, and crowding works through forgone contribution margin.

The proportional-appeal share has a discrete-choice interpretation when
capacity counts varieties or outlets and customers value nearby activities.
Appendix~\ref{app:logit} derives that interpretation following
\citet{train_discrete_choice_2009}, with opportunity value
combining market size and fixed per-sale margin.

\subsection{Costs and Equilibrium}\label{sec:model-equilibrium}

Firms differ in the activities they can expand cheaply.
A capability center $h_i\in\mathbb R^d$, a location in the activity space,
makes nearby activities less costly at the margin,
with the penalty for distance governed by $\gamma\geq0$.
The persistent capability cost is
\begin{equation}\label{eq:cost}
  C_i(q_i)=c_i^\top q_i+\tfrac12q_i^\top Hq_i,
  \quad c_{ia}=c_0+\gamma\|x_a-h_i\|^2/d,
  \quad H=\eta K^c+\nu I,
\end{equation}
where $c_0,\eta\geq0$, $\nu>0$, and $K^c_{ab}=\exp\{-\tau_c\|x_a-x_b\|^2/d\}$.
Dividing squared distances by $d$ expresses both distance penalties per
activity-space coordinate.
The linear term represents persistent comparative advantage, the Gaussian
congestion term makes simultaneous expansion in related activities costly,
and $\nu$ raises the marginal cost of concentrating capacity at any one site.
Increasing $\gamma$ makes departures from a firm's capability center more
expensive, so capability centers give firms different costs of responding to
the same opportunity.
These are costs of operating the allocation, including when it is unchanged.

Together, the proportional-appeal margin and the cost specification define
the operating payoff in \eqref{eq:firm-problem}.
The common Gaussian matrix in \eqref{eq:cost} is the numerical specialization
used below.
For the equilibrium argument, we allow a symmetric firm-specific matrix $H_i$
subject to the following conditions.

\begin{assumption}[Demand access and expansion costs]\label{ass:primitives}
  Outside appeal satisfies $b_\ell>0$, opportunity values and appeal
  coefficients
  satisfy $r_\ell\geq0$ and $K_{\ell a}\geq0$, and linear costs
  satisfy $c_{ia}\geq0$.
  Each symmetric cost matrix obeys $H_i\succeq\nu I$ for a common $\nu>0$.
\end{assumption}

The Gaussian congestion matrix is positive semidefinite, so the numerical
specification satisfies the curvature condition.
At equilibrium, marginal contribution equals marginal capacity cost on
every positive-capacity activity, while entry into an unused activity
cannot raise payoff.
To establish that these joint conditions determine a unique allocation,
Appendix~\ref{app:equilibrium} adapts the contest and variational-inequality
methods of
\citet{xu_zenou_zhou_conflict_2022} and
\citet{parise_ozdaglar_network_games_2019}.

\begin{proposition}[Unique equilibrium and the symmetry
  boundary]\label{prop:equilibrium}
  Under Assumption~\ref{ass:primitives}, the proportional-share capacity game
  has a unique Nash equilibrium on $\mathcal Q$.
  Identical firm primitives imply identical equilibrium capacity vectors.
  The independent-firm economy, obtained by setting $B_{i\ell}=b_\ell$,
  also has a unique equilibrium, consisting of the firms' individual optima.
\end{proposition}

Positive outside appeal keeps demand shares defined, while strictly curved
costs discipline scale and ensure uniqueness
(Appendix~\ref{app:equilibrium}).
Uniqueness also explains the symmetry boundary: crowding can spread aggregate
activity across opportunities, but firms with identical primitives spread it
identically.
Capability differences supply the heterogeneous incentives that separate
firms' business mixes in the illustrative economies.

\subsection{Sequential Adjustment and Valuation}\label{sec:model-dynamics}

Starting from the static equilibrium, firms incur additional costs when
changing their inherited capacity allocations.
In each period they choose capacities simultaneously, with adjustment costs
relative to the preceding allocation $q_{i,t-1}$.
The stage choice is
\begin{equation}\label{eq:anchored-choice}
  q_{it}\in\arg\max_{z_i\geq0}\left\{
  \pi_i(z_i,q_{-i,t};r_t)-\frac{\kappa}{2}\|z_i-q_{i,t-1}\|^2\right\},\qquad\kappa\geq0.
\end{equation}
Adjustment costs make operating responses gradual, as in the costly
repositioning problem studied by \citet{sweeting_positioning_2013}.
The quadratic form makes every stage a strongly monotone Nash problem
(Appendix~\ref{app:anchored-stage}).
Firms maximize current operating payoff net of adjustment expenditure,
without a continuation-value term.
Financial prices reflect the resulting choices and do not feed back into
investment.

At an interior initial equilibrium $\bar q$, let $J$ be the joint Jacobian
of marginal capacity costs minus marginal contributions, and let $\mathcal G$
map opportunity changes into changes in marginal incentives.
A local capacity deviation $x_t$ generated by opportunity deviation
$s_t v$ obeys
\begin{equation}\label{eq:anchored-response}
  (J+\kappa I)x_t=\kappa x_{t-1}+\mathcal G v s_t.
\end{equation}
Adjustment costs make earlier investment decisions persist: a firm responds
to new opportunities from the capacity allocation it has already built.
Because $J$ includes rivals' effects on marginal returns, competitors'
simultaneous adjustments also shape that response.
At a boundary equilibrium, equation~\eqref{eq:anchored-response} applies
locally while the active set remains unchanged.
The nonlinear simulation solves the actual level game in
\eqref{eq:anchored-choice} and checks its departure from this affine
response.

We value the cash flows generated by this adjustment path using a constant
discount factor $\delta\in(0,1)$.
With baseline payout $\bar\pi_i>0$, the ex-payout price is
$\bar V_i=\delta\bar\pi_i/(1-\delta)$.
For a unit permanent opportunity innovation, let $x_0=0$ and let $x_h$ for
$h\geq1$ solve \eqref{eq:anchored-response} with $s_h=1$.
The unexpected total-return loading is
\begin{equation}\label{eq:same-path-value}
  g_i=\frac{1}{\bar V_i}\sum_{h=0}^{\infty}\delta^h
  \left\{\partial_r\pi_i(\bar q;\bar r)v+D_q\pi_i(\bar q;\bar r)x_h\right\}.
\end{equation}
The first term values the opportunity at existing capacities; the second
values the cash-flow consequences of subsequent operating responses.
Investors can incorporate this anticipated path when the news arrives,
before capacity reaches its new allocation.
Adjustment expenditure has zero first derivative at the unchanged origin,
so it adds no first-order payout term to this valuation.

The discounted capacity sum solves
$[J+(1-\delta)\kappa I]\sum_h\delta^h x_h
=\delta\mathcal Gv/(1-\delta)$, valuing the full response without truncating
its horizon.
In the spirit of \citet{berk_green_naik_growth_options_1999}, operating
opportunities enter firm valuations through the cash flows they generate,
which here follow the supplied sequence of Nash choices under
constant discounting.

\subsection{Reporting and the Directional Prediction}\label{sec:model-observation}

Financial news records firms' activities through the coverage they receive.
Investors seeking to value future cash flows have reason to follow material
investment, product and technology decisions, while publishers and corporate
disclosures influence which developments become observable
\citep{gentzkow_media_bias_2006}.
To connect this reporting process to operating decisions, we assume that an
activity's share of news is proportional to its capacity weighted by its
reporting intensity.
Capacity measures the activity available to be reported, while reporting intensity
allows coverage per unit of capacity to differ across activities and time.
For fixed component probability laws $\nu_a$, this gives
\begin{equation}\label{eq:reporting-law}
  C_{i,t}=\sum_a\omega_{ia,t}\nu_a,\qquad
  \omega_{ia,t}=\frac{\lambda_{ia,t}q_{ia,t}}
  {\sum_b\lambda_{ib,t}q_{ib,t}},\qquad\lambda_{ia,t}>0.
\end{equation}
Common intensity across activities gives $\omega_{ia,t}=p_{ia,t}$, so the
article law follows the firm's business mix under equal relative attention.
Changing relative attention changes reported composition even at fixed
capacities, whereas reciprocal changes in capacity and intensity can
preserve the same observed law.
Text-based product measures \citep{hoberg_phillips_2016} provide a precedent
for studying activities through text, while the reporting weights specify
the maintained connection to operating capacity here.
The annual observation pools the dated article laws before evaluating a
transport distance.

In the restricted two-opportunity economy below, an improvement in one
opportunity relative to the other locally raises each firm's capacity in the
improving activity and lowers its capacity in the other.
Every firm therefore increases its share of the improving opportunity,
although the joint response includes the capacity choices of its rivals.
Firms with greater initial exposure to the improving opportunity have larger
return responses to the shock.
Under common reporting across activities, a firm locally increasing its share
of the improving opportunity approaches the origin profiles of firms with
larger shares of that opportunity and recedes from those with smaller shares.

Consider diversified firms with distinct interior shares $p_i$ in that
opportunity, common total capacity $Q$ and $\sum_i p_i=N/2$.
With identity activity loading, common baseline opportunity value $v$ and
denominator $D=b+NQ/2>2Q$, each node's response Jacobian is
$J=d_\perp I+u\mathbf1^\top$, where
$d_\perp=\nu+v/D^2$ and $u_i=v(D-2q_i)/D^3>0$.
The anchored response has two persistence modes, $\kappa/(d_\perp+\kappa)$
and $\kappa/(d_A+\kappa)$ with $d_A=d_\perp+\sum_i u_i$.

Valuing the cash flows along this adjustment path gives the return response
$g(p)=G(2p-1)/[1+(2p-1)^2]$ for $G>0$.
The return response $g(p)$ increases with the firm's initial exposure $p$.

For the unit-axis component laws
$\nu_-=(\delta_{e_1}+\delta_{-e_1})/2$ and
$\nu_+=(\delta_{e_2}+\delta_{-e_2})/2$, every mixture has zero centroid, yet
\begin{equation}\label{eq:axis-observation}
  W_2^2(C(p),C(q))=2|p-q|.
\end{equation}
The target therefore approaches the initial news distributions of reference
firms with larger return responses to the shock.
Under the timing and observation conditions below, this implies a nonnegative
expected local covariance between approach and relative returns.
Write $A_{ijy}$ for the local decrease in population Wasserstein squared
distance of target $i$ toward
reference $j$ over year $y$, and $R_{ijy}$ for the corresponding article-weighted
excess-return measure, with observed counterparts defined in
Section~\ref{sec:adjustment-design}.

\begin{proposition}[Directional movement--return
  alignment]\label{prop:directional-alignment}
  Maintain the interior two-opportunity construction of
  Appendix~\ref{app:canonical-mechanism} and the unit-axis observation laws
  under common reporting across activities.
  Hold origin profiles and article timing fixed, with nonnegative source-year
  timing weights and no earlier-window articles strictly after a current-year
  source.
  Suppose primitive source-date innovations are centered, independent across
  dates and have finite variances, and maintain the zero cross-moment
  restrictions on financial proxy error, initial gaps and observation error.
  Then, for each target $i$ and destination year $y$,
  \[
    E[\operatorname{Cov}_{j\ne i}(R_{ijy},A_{ijy})]\geq0.
  \]
  Fixed nonnegative pooling of these moments preserves the sign.
\end{proposition}

Appendix~\ref{app:canonical-mechanism} proves the result by deriving the
ordered share and valuation responses from the same joint operating path.
News may be shared across firms: independence concerns source dates, while
firms' different exposures generate their relative responses.

\subsection{Directional Measurement and Estimator}\label{sec:adjustment-design}

The model's exposure ordering concerns a target's movement relative to other
firms, which we measure by holding their article distributions fixed while
the target's distribution changes.
This preserves a common origin for each comparison, separating the target's
approach from the reference firm's contemporaneous movement.
Let $\widehat C_{iy}$ be firm $i$'s annual distribution over unit-normalized
article vectors and keep every reference law $\widehat C_{j0}$ at the 2018
origin, and for reference $j\ne i$ define
\begin{equation}\label{eq:frozen-approach}
  D_{ijy}=W_2^2(\widehat C_{iy},\widehat C_{j0}),\qquad
  A_{ijy}=D_{ij,y-1}-D_{ijy}.
\end{equation}
Positive $A$ means that target $i$ approaches reference $j$, with $A$ measured
in squared embedding-distance units under the squared Euclidean ground cost.
Under the two-opportunity law \eqref{eq:axis-observation}, changes in
distribution shape reveal movement even though all centroids coincide.

News of a free cash flow shock arriving early in a year can precede more of
the target's coverage than news arriving near its end.
We reflect this timing by weighting each reference return by the fraction
of target articles still to appear.
For return date $s$ in destination year $y$, denote that fraction by
$h_{iys}$ and the reference's daily simple decimal excess return by
$r^e_{js}$, and set
\begin{equation}\label{eq:dated-financial}
  h_{iys}=\frac{1}{n_{iy}}\sum_{\ell=1}^{n_{iy}}
  \mathbf1\{d_{iy\ell}>s\},\qquad
  R_{ijy}=\sum_{s\in y}h_{iys}r^e_{js}.
\end{equation}
We sum each reference firm's daily excess returns, weighting them by the
fraction of the target firm's articles published later in the year.
This is an additive return measure, not a compounded annual return.
The strict inequality excludes same-day news without assigning an intraday
arrival time.
Because dates remain paired with articles, each redraw changes both the
information distribution and its article-timing weights.

Centering subtracts the cross-reference mean within each target-year,
so the coefficient compares which reference profiles the target approaches
more closely with which references have higher weighted excess returns.
References receive equal weight within each target-year.
For $V=A,R$, write $\bar V_{iy}=(N-1)^{-1}\sum_{j\ne i}V_{ijy}$ and
$\widetilde V_{ijy}=V_{ijy}-\bar V_{iy}$.
With $Y$ transitions, the annual moments and pooled estimator are
\begin{equation}\label{eq:directional-ratio}
  q_y=\frac{\sum_i\sum_{j\ne i}\widetilde R_{ijy}\widetilde A_{ijy}}{N(N-1)},\quad
  d_y=\frac{\sum_i\sum_{j\ne i}\widetilde R_{ijy}^2}{N(N-1)},\quad
  \widehat\beta=\frac{\sum_y q_y}{\sum_y d_y}.
\end{equation}
The coefficient solves the aggregate centered moment $\sum_y(q_y-\beta
d_y)=0$ and minimizes the pooled squared residual after target-year intercepts.
A return date with zero article-timing weight contributes zero to the
weighted excess-return measure without invalidating the comparison.
A zero pooled financial denominator makes the coefficient undefined.
When annual ratios differ, pooling weights years by $d_y$, giving reference
firms with greater relative financial dispersion greater denominator
influence despite equal reference weights.

Each firm enters many comparisons, both as a target and as a reference.
Article redraws therefore reuse each sampled firm-year distribution wherever
that firm appears, preserving the resulting dependence.
These redraws assess article-selection variability for the observed firms
and financial histories under the scheme in Section~\ref{sec:data}.

Uncertainty under independently sampled complete firm histories concerns a
different target.
Appendix~\ref{app:annual-gmm} uses a complete-triple representation to
examine a variance approximation under independently sampled complete firm
histories, retaining dependence across years and both roles.
The simulation evaluates the two procedures against their respective known
targets.

\subsection{Restricted-Form Hypothesis}\label{sec:hypothesis}

In the restricted two-opportunity economy, when one opportunity improves
relative to the other, firms with greater initial exposure to it have more
favorable return responses, while other firms expanding toward it approach those
firms' initial news distributions.
Because Proposition~\ref{prop:directional-alignment} signs the moments that
\eqref{eq:directional-ratio} pools, its conditions imply a nonnegative population
ratio of expected pooled moments whenever the denominator is positive.
In the simulated economy, where this coefficient is strictly positive in every
history under common reporting, we therefore state the hypothesis as
$H_0:\beta\leq0$ against $H_1:\beta>0$.

Since changes in reporting alone can also produce a positive coefficient at
fixed capacities, we assess the estimator against a positive operating target
under common reporting, a zero target under stationary reporting, and a positive
reported target without operating movement.
In the corporate archive, where the proposition's conditions are not verified
and no established design samples independent firm histories, we report the
estimate with its sensitivity to article selection
rather than test $H_0$.

\section{Application}\label{sec:application}\label{sec:simulation}

To assess how reliably finite samples of financial news recover the positive association
predicted by the model, we develop a multi-agent simulation of the economy in
Section~\ref{sec:model}, whose population coefficients are known.
We examine the effects of article coverage and changes in reporting on these
estimates, then apply the estimator to corporate news in
Section~\ref{sec:data}.

\subsection{Experimental Design}\label{sec:simulation-design}

The experiment follows competitive adjustment after small exogenous free
cash flow shocks,
keeping operating movements close to the local setting of the directional
prediction.
Sampling articles along these paths reveals how well the estimator detects
small reallocations at the coverage observed in the corporate data.

The competitive economy has 100 firms with canonical interior baseline
shares $p_i=i/101$, total
capacity one, two activities and two opportunities.
Outside appeal and baseline opportunity levels equal one, with cost curvature
$\nu=1/(4D^2)$, where $D=51$, linear costs supporting the interior origin,
and adjustment weight $\kappa=d_\perp=\nu+1/D^2$.

The 60 monthly ticks begin from stationary conditions, which we label 2018 for
comparison with the corporate sample, followed by independent
permanent innovations of $\pm0.001$ at ticks 18, 30, 42 and 54 that raise one
opportunity and lower the other.
We evaluate all sixteen sign histories and assign each equal weight.
At monthly discounting $\delta=0.95^{1/12}$, we compute return responses using
the infinite permanent-response resolvent in \eqref{eq:same-path-value}, while the
operating engine solves every contemporaneous joint Nash stage at the actual
positive opportunity levels.

In the first responding month after the first innovation, absolute changes
in normalized operating shares range from
\ppnum[4]{ir-redraw-signal-first-response-min-pct} to
\ppnum[4]{ir-redraw-signal-first-response-max-pct}
percentage points across firms and histories, with the largest absolute departure
from an origin share over the complete paths reaching
\ppnum[3]{ir-redraw-signal-maximum-origin-change-pct} percentage points.
These small movements specify the local signal tested in this economy.

The baseline article budget uses the firm-year article counts and frozen
empirical monthly date distributions of the primary corporate sample
described in Section~\ref{sec:data}.
The next budgets multiply every count by four and sixteen.
For each history, \ppnum[0]{ir-validation-redraws} independent article panels
use nested count increments
across budgets and shared sampled monthly counts across reporting regimes.
The same dated counts determine annual article laws and article-timing
weights, and each sampled origin law is reused wherever it serves as a
reference.
Because multinomial month and atom counts can leave unequal masses on antipodal
atoms, we compute transport exactly on the axis support, which retains that
asymmetry.
Reusing each solved operating path across redraws isolates article
measurement variation around the same economic response.

For each history $h$, we distinguish approximation of the operating response
from its measurement through articles.
Comparing the affine-share coefficient with the finite-capacity
common-reporting coefficient $\beta_h^0$ assesses the local economic
approximation.
The population reporting-law coefficient $\beta_h^\lambda$ incorporates
changing attention, while the sampled estimate
$\widehat\beta_{h,r,s}^\lambda$ adds finite-article variation at budget $s$.
Table~\ref{tab:directional-validation-controls} reports absolute error
against each history's population-law target and the frequency of numerically
resolved positive estimates, averaging histories equally.
Because every common-reporting population coefficient is positive, its
positive-estimate frequency also measures correct-sign recovery.
The fixed draw count targets a stratified aggregate rate standard error below
one percentage point.
Appendix~\ref{app:recovery-diagnostics} specifies the arithmetic sign allowance
and retains the full directional and magnitude diagnostics.

\subsection{Estimator Recovery and Reporting Effects}\label{sec:simulation-recovery}

% AUTO-GENERATED from saved Paper 6 results; do not edit.
\begin{table}[htbp]
\centering
\caption{Mechanism and measurement controls. Article-count multipliers scale observed coverage. Mean target is the arithmetic mean of history-specific population coefficients. Positive reports numerically resolved positive estimates as percentages (stratified Monte Carlo SE in percentage points). Absolute RMSE uses each history's matching target in coefficient units. Stationary reporting is the zero-alignment measurement control; reporting only holds capacities fixed and uses reporting amplitude 0.005. Table~\ref{tab:directional-validation-details} provides full diagnostics and definitions. Each percentage uses \ppnum[0]{ir-validation-common-1-positive-denominator} attempted redraws.}
\label{tab:directional-validation-controls}
\begin{tabular}{@{}l@{\enspace}r@{\enspace}r@{\enspace}r@{\enspace}r@{}}
\toprule
Design & \shortstack{Article-count\\multiplier} & Mean target & \shortstack{Positive estimates\\(\%)} & Absolute RMSE \\
\midrule
Common reporting & $1\times$ & \ppnum[3]{ir-validation-common-1-mean-target} & \ppnum[1]{ir-validation-common-1-positive-pct}(\ppnum[2]{ir-validation-common-1-positive-se-pct}) & \ppnum[3]{ir-validation-common-1-rmse} \\
Common reporting & $4\times$ & \ppnum[3]{ir-validation-common-4-mean-target} & \ppnum[1]{ir-validation-common-4-positive-pct}(\ppnum[2]{ir-validation-common-4-positive-se-pct}) & \ppnum[3]{ir-validation-common-4-rmse} \\
Common reporting & $16\times$ & \ppnum[3]{ir-validation-common-16-mean-target} & \ppnum[1]{ir-validation-common-16-positive-pct}(\ppnum[2]{ir-validation-common-16-positive-se-pct}) & \ppnum[3]{ir-validation-common-16-rmse} \\
Stationary reporting & $1\times$ & \ppnum[3]{ir-validation-stationary-1-mean-target} & \ppnum[1]{ir-validation-stationary-1-positive-pct}(\ppnum[2]{ir-validation-stationary-1-positive-se-pct}) & \ppnum[3]{ir-validation-stationary-1-rmse} \\
Reporting only & $1\times$ & \ppnum[3]{ir-validation-reporting-only-mean-target} & \ppnum[1]{ir-validation-reporting-only-positive-pct}(\ppnum[2]{ir-validation-reporting-only-positive-se-pct}) & \ppnum[3]{ir-validation-reporting-only-rmse} \\
\bottomrule
\end{tabular}
\end{table}

Under common reporting, correct-sign recovery rises from
\ppnum[1]{ir-validation-common-1-correct-pct}\% at observed coverage to
\ppnum[1]{ir-validation-common-4-correct-pct}\% at fourfold coverage and
\ppnum[1]{ir-validation-common-16-correct-pct}\% at sixteenfold coverage.
Even the largest article samples therefore often fail to recover a
numerically resolved positive sign, despite a normalized difference of only
\ppnume[2]{ir-finite-affine} between the finite-capacity and affine-share
population coefficients.

The small operating-share changes shift population article weights only
slightly, so finite-article variation can obscure the corresponding movement
toward reference profiles.
Coefficient error also falls with article coverage, but remains large
relative to the population target.
Normalized RMSE under common reporting is
\ppnum[2]{ir-validation-common-1-normalized-rmse},
\ppnum[2]{ir-validation-common-4-normalized-rmse} and
\ppnum[2]{ir-validation-common-16-normalized-rmse} across the same ladder.
Even at sixteenfold coverage, error remains several times the root-mean-square
magnitude of the population coefficients.

Reporting intensity can also respond to the opportunity itself.
We set $\lambda_{\pm,t}=\exp(\pm0.005u_t)$, where $u_t$ is cumulative
news divided by $0.004$, its maximum possible magnitude, so attention shifts
toward the activity whose contribution-margin opportunity has increased.
At sixteenfold coverage, correct-sign recovery relative to the operating-law
target rises to
\ppnum[1]{ir-validation-changing-16-correct-pct}\%
(Appendix~\ref{app:recovery-diagnostics}).
At fixed capacities, reporting changes alone generate a positive population
coefficient in all
\ppnum[0]{ir-reporting-only-positive-histories} histories.
The reporting-only row in Table~\ref{tab:directional-validation-controls}
measures this reported-law target, while its operating-law target is zero.
The counterfactual retains return responses from the modeled operating
path and changes only the observation law, without solving a separate
fixed-capacity economy.

Holding the reported population laws stationary under common reporting
isolates article noise around a zero-alignment target.
We retain the dated financial process, so population approach is exactly zero
while the financial denominator remains positive.
Finite articles nevertheless yield coefficients of both signs, with positive
frequencies near one half across the coverage ladder.
Because a zero target supplies neither a correct positive sign nor a scale
for normalized RMSE, we assess this measurement control through absolute
errors and sign frequencies.
The control holds reporting laws fixed rather than solving a new competitive
equilibrium.

\subsection{Measurement Uncertainty and Inference}\label{sec:simulation-uncertainty}

The spread of finite-article estimates provides a benchmark for assessing
whether redrawing articles from an observed archive represents that
measurement uncertainty.
We generate independent archives from the common-reporting, stationary and
reporting-only laws, retaining 100 firms and observed article coverage, then
apply the corporate article-redraw procedure within each archive.
Across \ppnum[0]{ir-redraw-outer-replications} archives per history, each
with \ppnum[0]{ir-redraw-inner-redraws} inner redraws, the broad redraw
ranges give near-nominal containment of the population coefficients.
Containment and detection answer different questions.
Under the positive-target common-reporting and reporting-only laws, the fifth
percentile exceeds zero in
\ppnum[1]{ir-redraw-mechanism-positive-q05-pct}\% and
\ppnum[1]{ir-redraw-reporting-positive-q05-pct}\% of attempts, respectively.
Thus, intervals that broadly capture measurement uncertainty rarely distinguish
these small positive population coefficients from zero.
Under stationary reporting, the target is zero, so the corresponding
\ppnum[1]{ir-redraw-stationary-positive-q05-pct}\% is instead a null-control
rejection frequency.

Uncertainty across independently sampled firms concerns a different target
from article selection around fixed operating paths.
For each target coefficient, we generate
\ppnum[0]{ir-inference-n100-replications} independent panels, each consisting
of iid firm histories.
Increasing the number of firms from 16 to 100 raises coverage of the
candidate one-sided 95\% lower bound from
\ppnum[1]{ir-inference-audit-n16-coverage-pct}\% to
\ppnum[1]{ir-inference-n100-coverage-pct}\%.
Zero-null rejection falls from
\ppnum[1]{ir-inference-audit-n16-null-rejection-pct}\% to
\ppnum[1]{ir-inference-n100-null-rejection-pct}\%
(Monte Carlo SE \ppnum[2]{ir-inference-n100-null-rejection-se-pct}
percentage points at 100 firms).
These generated-history panels preserve dependence across years and
comparison roles under the specified iid sampling law, although they are not
jointly solved competitive economies.
Appendix~\ref{app:annual-gmm} gives the sampling designs, target containment
and Monte Carlo errors.

\subsection{Empirical Design and Data}\label{sec:data}

To examine whether the association that the simulation recovers only imperfectly
also appears in corporate news, we apply the estimator of
Section~\ref{sec:adjustment-design} to \ppnum[0]{ir-firms}
Nasdaq-listed firms over five article years from 2018 through 2022, giving four
transitions ending in 2019--2022.

We measure semantic similarity between articles by embedding them with
Qwen3-Embedding-4B \citep{qwen3_embedding_2025}, a four-billion-parameter
transformer developed for information retrieval.
Because the model postdates the sample, these vectors give a retrospective
representation of the articles.
We assess the role of representation by comparing its native
\ppnum[0]{ir-native-dimension}-dimensional vectors with a compressed
64-dimensional representation.
As the model is trained with Matryoshka representation learning
\citep{kusupati_matryoshka_2022}, we obtain the compressed vectors from the first
64 coordinates of each native vector, renormalized to unit length.

We also compare two article samples within each representation to measure
sensitivity to the amount and balance of coverage.
While the expanded sample contains \ppnum[0]{ir-expanded-articles} existing
article embeddings, the matched sample restricts these to the first 128 per firm
in URL-hash order after the corpus date and body-length filters, giving
\ppnum[0]{ir-matched-articles} articles.
Taking the native expanded specification as primary, we treat the other three as
diagnostic comparisons, none selected by its financial coefficient.
All retain the same roster, origin references and annual financial windows,
but changing the article sample changes both geometry and the dates entering
\eqref{eq:dated-financial}.

For \ppnum[0]{ir-redraws} paired conditional article redraws, each firm-year
cloud is redrawn once within its canonical and additional-article strata and
reused in every target and reference role.
Origin draws remain shared across all transitions and representations, and
the return weights are recomputed from the sampled article dates.
The reported 2.5th, 5th and 97.5th percentiles describe article-selection
sensitivity conditional on the observed articles, firms, financial histories
and four annual environments.
Shared stories, common shocks and reused references make the
\ppnum[0]{ir-coordinates} directional coordinates dependent.

\section{Results}\label{sec:results}

\subsection{Historical Directional Evidence}\label{sec:historical-evidence}

Across all four measurement specifications, decreases in Wasserstein squared
distance are positively
associated with reference firms' relative weighted excess returns, and the central
article-redraw ranges remain positive
(Table~\ref{tab:directional-evidence}).
For a scale-free comparison within the primary specification, a difference of
one root-mean-square unit in centered weighted excess returns corresponds to
a fitted decrease in distance equal to \ppnum[3]{ir-normalized-alignment}
times the root-mean-square variation in centered distance changes.
This normalized alignment,
$\widehat\beta\operatorname{RMS}(\widetilde R)/\operatorname{RMS}(\widetilde A)$,
has square \ppnum[4]{ir-d2560-expanded-within-fit}, the in-sample fit of this
single-regressor, within-transformed projection.
It describes a small alignment in the observed sample.

For the primary specification, the coefficient is
\ppnum[6]{ir-primary-beta} squared-distance units per unit additive
weighted excess-return measure, with conditional article-redraw range
$[\ppnum[6]{ir-primary-low},\ppnum[6]{ir-primary-high}]$.
A \ppnum[2]{ir-financial-example} difference in the weighted excess-return measure
corresponds to a fitted approach
difference of \ppnume[2]{ir-primary-percent-effect} squared-distance units.
Because the coefficient describes direction in embedding geometry, its scale
differs across representations and has no percentage-capacity interpretation.

% AUTO-GENERATED from saved Paper 6 results; do not edit.
\begin{table}[htbp]
\centering
\caption{Complete-sample historical directional association. Native denotes 2,560 normalized coordinates. Percentiles describe conditional article-resampling sensitivity under the scheme in Section~\ref{sec:data}. Alignment is the normalized pooled association; its square is descriptive within-projection fit. All measurement specifications use 100 firms, fixed 2018 reference profiles and four transitions ending in 2019--2022.}
\label{tab:directional-evidence}
\begin{tabular}{lrrrrr}
\toprule
Specification & $\widehat\beta$ & \shortstack{2.5th\\percentile} & \shortstack{5th\\percentile} & \shortstack{97.5th\\percentile} & Alignment \\
\midrule
Native, expanded & \ppnum[6]{ir-d2560-expanded-beta} & \ppnum[6]{ir-d2560-expanded-low} & \ppnum[6]{ir-d2560-expanded-q05} & \ppnum[6]{ir-d2560-expanded-high} & \ppnum[4]{ir-d2560-expanded-alignment} \\
Native, matched & \ppnum[6]{ir-d2560-matched-beta} & \ppnum[6]{ir-d2560-matched-low} & \ppnum[6]{ir-d2560-matched-q05} & \ppnum[6]{ir-d2560-matched-high} & \ppnum[4]{ir-d2560-matched-alignment} \\
Leading 64, expanded & \ppnum[6]{ir-d64-expanded-beta} & \ppnum[6]{ir-d64-expanded-low} & \ppnum[6]{ir-d64-expanded-q05} & \ppnum[6]{ir-d64-expanded-high} & \ppnum[4]{ir-d64-expanded-alignment} \\
Leading 64, matched & \ppnum[6]{ir-d64-matched-beta} & \ppnum[6]{ir-d64-matched-low} & \ppnum[6]{ir-d64-matched-q05} & \ppnum[6]{ir-d64-matched-high} & \ppnum[4]{ir-d64-matched-alignment} \\
\bottomrule
\end{tabular}
\end{table}

All annual and deletion-based estimates remain positive.
The primary annual ratios are \ppnum[6]{ir-annual-2019-beta},
\ppnum[6]{ir-annual-2020-beta}, \ppnum[6]{ir-annual-2021-beta} and
\ppnum[6]{ir-annual-2022-beta} in 2019--2022, while leaving out one
transition gives coefficients from
\ppnum[6]{ir-year-deletion-low} to \ppnum[6]{ir-year-deletion-high}.

Removing each firm in all target and reference roles, then recentering over
remaining references, gives the primary range
$[\ppnum[6]{ir-role-deletion-low},\ppnum[6]{ir-role-deletion-high}]$.

These comparisons contribute unevenly to the pooled estimate:
the 2020 transition contributes \ppnum[1]{ir-primary-2020-weight-pct}\%
of the pooled second moment, and the largest reference share is
\ppnum[1]{ir-reference-share-pct}\%, compared with
\ppnum[2]{ir-target-share-pct}\% for a target.
A reference with unusually dispersed weighted returns can affect many target
comparisons, so equal reference weights need not imply equal influence on
the pooled slope.

The restricted movement profile improves held-out reconstruction of annual
news changes by only \ppnum[2]{ir-observation-primary-heldout-pct}\%
relative to predicting zero centered change, and two alternative
specifications perform worse than that benchmark.
This article-fold audit examines changes in news geometry separately from
their association with returns.
With the profile and reference ordering fixed from 2018 geometry, it fits an
annual movement magnitude on one article fold and evaluates it on the other.
Weak held-out reconstruction shows that annual article changes extend beyond
this local benchmark.
The measured movements also frequently reach or exceed the distance to the
nearest distinct reference distribution at the initial date.
These geometric departures do not identify crossings of operating positions,
which remain unobserved.
Appendix~\ref{app:observation-audit} reports the complete audit.

\subsection{Discussion}\label{sec:discussion}

The findings provide an economic interpretation of how firms' economic
footprints can evolve in response to changing market opportunities.
Distribution-valued representations allow researchers to compare firms
through the information associated with their activities, but changes in
those distributions do not, by themselves, explain the economic decisions
underlying repositioning.
Our model connects these movements to strategic competition.
When firms adjust capacity across overlapping activities, their decisions
affect competitors' margins, while adjustment costs make the response to new
opportunities gradual.
In the restricted two-opportunity economy, initial operating exposures
determine both the ordering of return responses and the direction of
movement relative to competitors' initial news distributions.
This establishes a mechanism through which competitive adjustment can
generate systematic changes in firms' reported economic footprints, extending
distribution-based comparisons beyond static measures of similarity.

The distinction between financial valuation and operating adjustment is
central to this mechanism.
Investors can capitalize the anticipated cash-flow consequences of an
exogenous free cash flow shock when information arrives, whereas firms adjust their
operating allocations over subsequent periods.
Return responses and changes in reported economic footprints consequently need not
coincide in time.
The article-timing weights in our estimator connect reference firms' daily
excess returns to the subsequent appearance of information about the
adjusting firm.
This construction captures a directional relationship between weighted
excess returns and changes in the firm's reported economic footprint, rather
than identifying a structural adjustment speed.
Its economic interpretation depends jointly on initial exposures, competitive
interactions, adjustment costs, and the timing of information.

The historical results reveal a positive but economically small association.
The estimated direction is consistent across four measurement specifications,
and the primary normalized association is \ppnum[4]{ir-normalized-alignment},
corresponding to approximately
\num[round-mode=places,round-precision=2]{\fpeval{100*\ppvalue{ir-d2560-expanded-within-fit}}}\%
descriptive within-sample fit.
Thus, the predicted ordering appears in the data but accounts for only a
limited component of observed information movements.
This modest magnitude is compatible with gradual operating adjustment,
heterogeneous responses to economic opportunities, and imperfect measurement
of firms' activities.
The simulations clarify the last possibility: at the benchmark's small
adjustment scale, finite article samples frequently fail to resolve a
positive population relationship, even when coverage increases substantially.
They do not, however, establish which mechanism explains the small historical
association.
The evidence supports a directional correspondence without demonstrating that
the underlying operating response is either strong or uniform across firms.

The reconstruction results sharpen this interpretation.
In the primary specification, the restricted movement profile improves
held-out reconstruction of annual news changes by only
\ppnum[2]{ir-observation-primary-heldout-pct}\% relative to a zero-change
benchmark, while two alternative specifications perform worse than that benchmark.
A positive directional association can therefore coexist with weak
reconstruction of the broader evolution of firms' reported economic footprints.
These findings answer different questions.
The directional estimate measures whether relative movements toward reference
firms correspond to their relative weighted excess returns, whereas the reconstruction
exercise evaluates how well the model's restricted geometric pattern
describes observed changes in news distributions.
The theory predicts the former under specified conditions, not a complete
account of the latter.
Nevertheless, the weak reconstruction indicates that the historical movements
contain substantial variation beyond the local theoretical approximation.
The empirical evidence is consequently more informative about relative
directional ordering than about the full geometry or magnitude of firms' repositioning.

An additional implication concerns the production of financial information itself.
Corporate news reflects both underlying economic activity and the attention
directed toward different activities.
Under the observation model, changes in reporting intensity alter a firm's
information distribution even when its operating allocation remains unchanged.
The reporting-only simulation demonstrates that this channel can generate
positive directional alignment at fixed capacities.
Operating adjustment and reporting responses can therefore produce similar
observed movements, potentially reinforcing or offsetting one another.
This distinction is not simply a source of statistical noise: it affects the
economic meaning of the distances used to represent firms.
The historical estimator measures changes in reported economic footprints
associated with relative weighted excess returns, but does not separately identify the
contributions of operating decisions and reporting intensity.
Interpreting information-space movements as changes in economic activity
requires assumptions about how that activity becomes observable.

The stability of the estimated sign provides useful evidence about the
historical association, although it does not establish its generality.
The primary coefficient remains positive across annual transitions and the
reported deletion exercises.
At the same time, the 2020 transition contributes
\ppnum[1]{ir-primary-2020-weight-pct}\% of the pooled centered financial
second moment, and the largest reference-firm contribution accounts for
\ppnum[1]{ir-reference-share-pct}\%.
The financial variation informing the estimate is therefore unevenly
distributed across the sample.
Moreover, the theoretical prediction arises under restricted opportunity and
observation conditions that need not characterize the wider corporate environment.
Stability of the estimated direction should accordingly be distinguished from
evidence that the same economic mechanism operates with comparable strength
across firms, activities, and market conditions.

Taken together, these results extend the economic interpretation of
distribution-valued firms by linking changes in reported economic footprints to
competitive adjustment and financial valuation.
Building on the multi-contest allocation framework of
\citet{xu_zenou_zhou_conflict_2022}, the analysis establishes how inherited
capacities, strategic responses, and anticipated cash flows can jointly
produce a testable directional relationship in information space.
The historical application identifies a corresponding association, while the
simulations and reconstruction exercises establish important boundaries on
its interpretation.
More direct tests would combine dated information about investment,
production, or operating capacity with financial returns and corporate news,
allowing changes in economic activity to be distinguished from changes in reporting.
Such evidence could also examine whether the predicted ordering extends
beyond the restricted setting considered here.
The broader contribution is an economic foundation for studying firms'
changing economic footprints through financial news.

\section{Conclusion}\label{sec:conclusion}

Firms reallocate capacity as market opportunities change, and their decisions
alter the margins available to competitors.
Adjustment costs make these responses gradual, while investors can value the
anticipated cash-flow consequences when news arrives.
Building on \citet{xu_zenou_zhou_conflict_2022}, we characterize sequential
Nash capacity choices, establish equilibrium and comparative-static results,
and connect firms' operating paths to financial valuation.
In a restricted two-opportunity economy, firms with greater initial exposure
experience larger return responses to positive exogenous free cash flow shocks.
As another firm increases its exposure to the improving activity, its news
distribution moves closer to
the initial distributions of more-exposed competitors when reporting
intensity is equal across activities.
Under the stated conditions, this generates a nonnegative expected local
association between relative financial returns and changes in firms'
positions in information space.

We investigate this prediction using squared Wasserstein distances between
distributions of corporate-news embeddings and excess returns weighted by
article timing.
Monte Carlo experiments show that limited article coverage makes the positive
population relationship difficult to recover at the small
operating-adjustment scale considered.
Changes in the relative coverage of activities can also produce a positive
association while capacities remain fixed.
In a sample of \ppnum[0]{ir-firms} Nasdaq-listed firms over 2018--2022, the
estimated association is positive across four specifications but small.
A separate held-out exercise finds that the restricted movement pattern
accounts for little of the variation in annual news distributions.

By treating firms' economic footprints as outcomes of strategic choices, the
paper adds an explanation for how those footprints evolve to research on
distribution-valued firms.
The analysis also establishes why the production of financial news matters
for measuring those changes: operating reallocations and shifts in coverage
can have similar effects on observed distributions.
Independent, dated measures of operating activity would allow future research
to distinguish these mechanisms and assess the model more directly.
This provides a basis for studying changes in firms' economic footprints
through their reported activities and financial returns.

\section*{Funding}

The authors received no financial support for the research, authorship, or
publication of this article.

\section*{Disclosure statement}

The authors report no potential conflict of interest.

\section*{Data availability statement}

The corporate articles come from an archive of Nasdaq financial news, and daily
returns come from commercial and public market-data providers.
These inputs remain subject to their providers' terms and are not redistributed.
The versioned code, environment specifications, derived article measures and
scripts that regenerate the manuscript's tables and figures are maintained in a
replication package that will accompany the version of record.

\section*{Declaration of generative AI use}

During the preparation of this work the authors used Codex~0.160.0 and
Claude Code~2.1.318 to assist with code refactoring and editorial review.
After using these tools, the authors reviewed and edited the content as needed
and take full responsibility for the content of the publication.

\clearpage
\bibliographystyle{plainnat}
\bibliography{refs}

\clearpage
\appendix
\section{A sufficient two-opportunity mechanism}\label{app:canonical-mechanism}

The construction in Section~\ref{sec:model-observation} uses $N\geq4$ firms,
$b,v,Q,\nu>0$, distinct shares $0<p_i<1$ with $\sum_i p_i=N/2$, and $D=b+NQ/2>2Q$.
Baseline capacities are $(Qp_i,Q(1-p_i))$.
At each node, maintain the additional restriction that the supporting linear
cost $c_i=v(D-q_i)/D^2-\nu q_i$ is nonnegative, and choose quadratic cost
$\nu q_i^2/2$, so the baseline marginal payoff vanishes.
The following derivation concerns this interior origin and its tangent
response, not unrestricted finite shocks.

For one node with baseline capacities $q_i$, differentiating the
marginal-cost-less-marginal-contribution field gives diagonal entry
$\nu+2v(D-q_i)/D^3$ and off-diagonal row entry $v(D-2q_i)/D^3$.
Consequently
\begin{equation}\label{eq:canonical-jacobian}
  J=dI+u\mathbf1^\top,\quad d=\nu+v/D^2,\quad
  u_i=v(D-2q_i)/D^3>0,\quad f_i=(D-q_i)/D^2.
\end{equation}
Here $f_i$ is the opportunity forcing.
Writing $U=\sum_i u_i=Nv(D-Q)/D^3$ and $F=\sum_i f_i=N(D-Q/2)/D^2$, define $d_A=d+U$.
The projectors $P_A=u\mathbf1^\top/U$ and $P_\perp=I-P_A$ obey $JP_A=d_A P_A$
and $JP_\perp=dP_\perp$.
Applying them to the anchored recurrence yields persistence modes
$a=\kappa/(d+\kappa)$ and $a_A=\kappa/(d_A+\kappa)$.
For a permanent unit node improvement, the capacity response after $h$ stages is
\begin{equation}\label{eq:canonical-step}
  x_h(q)=\frac{f(q)-(F/U)u(q)}{d}(1-a^h)
  +\frac{(F/U)u(q)}{d_A}(1-a_A^h),\qquad x_0(q)=0.
\end{equation}
This follows by summing the two geometric recurrences, with $0\leq a_A\leq a<1$.
The aggregate node response is $T_h=F(1-a_A^h)/d_A$.

The response $x_h(q)$ is affine in $q$ and positive on $[0,Q]$ for $h\geq1$.
At $q=Q$, both modal coefficients in \eqref{eq:canonical-step} are positive, since
$f(q)-(F/U)u(q)=(q-Q/2)/[D(D-Q)]$.
At $q=0$, write the response as
$A(1-a_A^h)-B(1-a^h)$, with
$A=(D-Q/2)/[D(D-Q)d_A]$ and $B=Q/[2D(D-Q)d]$.
Both constants are positive, and $A>B$ because
$2(D-Q)d-QU=2(D-Q)(\nu+vb/D^3)>0$.
Since $1-a_A^h\geq1-a^h>0$, this endpoint response is also positive; affinity
covers the interval.
A relative innovation raises the plus-node capacity by $x_h(Qp_i)$ and lowers
the minus-node capacity by $x_h(Q(1-p_i))$.
Differentiating the normalized plus share therefore gives
\begin{equation}\label{eq:canonical-share-step}
  \xi_i(h)=\frac{(1-p_i)x_h(Qp_i)+p_i x_h(Q(1-p_i))}{Q}>0
  \quad(h\geq1).
\end{equation}
The positive share response is generated by the joint game, including both modes.

The static envelope evaluates cash flow along exactly this capacity response.
Affinity of $x_h(q)$ and the balanced baseline shares imply that firm $i$'s
payout derivative is $(2p_i-1)L_h$, where
\begin{equation}\label{eq:canonical-payout-scale}
  L_h=\frac QD+\frac{vQ}{D^2}[x_h(Q)-T_h].
\end{equation}
This scale is strictly positive.
Indeed, $x_h(Q)\geq0$, $T_h\leq F/d_A$, and
$Dd_A-vF=D\nu+vb/D^2>0$, which bounds $L_h$ below by a positive constant.
The bounded stable responses make $L_\delta=\sum_{h\geq0}\delta^h L_h$ finite
and positive for $0<\delta<1$.
Substituting the supporting costs into baseline cash flow gives
$\bar\pi_i=B[1+(2p_i-1)^2]$, with $B=(v/D^2+\nu/2)Q^2/2>0$.
Dividing discounted payout news by the positive ex-payout price yields
\begin{equation}\label{eq:canonical-return-shape}
  g(p_i)=G\frac{2p_i-1}{1+(2p_i-1)^2},\qquad
  G=\frac{(1-\delta)L_\delta}{\delta B}>0.
\end{equation}
Its derivative is $2G[1-(2p-1)^2]/[1+(2p-1)^2]^2>0$ on $0<p<1$.
Thus the financial ordering follows from the same joint operating response
and constant-discount valuation convention, rather than an independently
imposed return ranking.

For the observation laws in \eqref{eq:axis-observation}, common mass can stay
at each antipodal atom; surplus component mass must move to the other axis at
squared cost two per unit.
This constructs cost $2|p-q|$, and the required surplus provides the matching
lower bound for every coupling.
At distinct baseline shares, approach toward reference $j$ has derivative
$2o_{ij}$ times the target share change, with $o_{ij}=\operatorname{sgn}(p_j-p_i)$.
Pooling dated probability laws before evaluating transport preserves the
corresponding local weighted share derivative.

Let $f_{iys}\geq0$ be the article-timing weight of source date $s$ and
$k_{iys}\geq0$ the later-year article-weighted response assembled from
\eqref{eq:canonical-share-step}.
For sources in year $y$, the earlier article window contains no observations
after the source; strict timing excludes same-date financial news.
With centered independent primitive innovations $\varepsilon_s$, finite
variances $\sigma_s^2$, and no cross moments with initial gaps or
proxy/observation errors, the tangent centered moment is
\begin{equation}\label{eq:canonical-directional-moment}
  E[\operatorname{Cov}_{j\ne i}(R_{ijy},A_{ijy})]
  =2\operatorname{Cov}_{j\ne i}(g(p_j),o_{ij})
  \sum_{s\in y}f_{iys}k_{iys}\sigma_s^2\geq0.
\end{equation}
To see the last sign, express the reference covariance as half the equally
weighted ordered pair sum of
$[g(p_j)-g(p_k)](o_{ij}-o_{ik})$.
Both factors have the same ordering in baseline share, so every product is nonnegative.
Independence across source dates removes cross-source expectations; primitive
news can be shared by every firm.
Nonnegative pooling across target-years preserves the sign.
General response and financial loading matrices need not have this ordering,
and observed raw returns and reports need not satisfy the exclusion restrictions.
The aligned operating and return responses therefore give a sufficient
local mechanism for the expected moment in
Proposition~\ref{prop:directional-alignment}.

\section{Static response and exposure support}\label{app:static-response}

\subsection{Direct and full responses}\label{sec:model-counterfactuals}

Starting from the strategic equilibrium
$q^0=q^{\mathrm S,*}(r_0,c_0)$, we compare responses to the same change in
opportunity values.
With rivals frozen at $q_{-i}^0$, each firm solves its new capacity problem;
the collection of these individual best responses is the direct allocation
and need not itself be an equilibrium.
Allowing rivals to respond instead gives the full allocation, the new
strategic equilibrium.
Because both calculations share the initial appeal, saturation and marginal
costs at $q^0$, full minus direct isolates rival adjustment at that origin.
The same definitions apply to linear-cost changes with opportunities fixed.
Removing rival appeal altogether gives a different benchmark: its
independently equilibrated origin changes as well as its response
(Appendix~\ref{app:independent-proof}).

\subsection{Comparative statics of rival feedback}\label{sec:model-response}

To characterize the common-origin comparison locally, first change
opportunity values while holding all cost primitives fixed.
At an equilibrium marginal contribution equals marginal cost on active
coordinates, and a small disturbance must preserve that balance.
On firm $i$'s active coordinates this requires
\begin{equation}\label{eq:firm-response}
  \underbrace{J_{ii}\,\delta q_i}_{\text{own adjustment}}
  =\underbrace{\mathcal G_i\,\delta r}
  _{\substack{\text{changed opportunity}\\\text{incentives}}}
  -\underbrace{\sum_{j\ne i}J_{ij}\,\delta q_j}
  _{\substack{\text{changed incentives}\\\text{from rivals}}}.
\end{equation}
Here $J_{ii}$ measures how quickly firm $i$'s incentive to expand declines as
it adds capacity, $J_{ij}$ measures how rival $j$'s capacity changes that
incentive, and $\mathcal G_i$ translates changes in node opportunities into
changes in the marginal returns to firm $i$'s activities.
The frozen-rival response sets every rival movement $\delta q_j$ to zero; the
full response solves all firms' equations jointly.

Write $J$ for the Jacobian of marginal cost less marginal contribution,
$J_{\mathrm d}=\operatorname{blockdiag}(J_{ii})$ for its own blocks and
$J_{\mathrm o}=J-J_{\mathrm d}$ for its rival blocks.
Let $\mathcal U$ collect positive-capacity coordinates.
Strict complementarity means that every unused activity has a strictly
negative incentive to enter, so small disturbances leave this active
set unchanged.
All inverses below are restricted to $\mathcal U$;
Appendix~\ref{app:response} supplies their entries and the projected systems.
Theorem~3 of \citet{xu_zenou_zhou_conflict_2022} gives the corresponding
active-coordinate comparative statics in their pure-cost setting under their
cost assumptions and nondegeneracy.
Here the same fixed-active-set logic supports an explicit comparison between
joint adjustment and the own-firm response with rivals frozen.

\begin{proposition}[Equilibrium response and common-origin
  feedback]\label{prop:response}
  Maintain Assumption~\ref{ass:primitives} at an equilibrium with strict
  complementarity.
  For an opportunity direction $\delta r$, the full equilibrium derivative
  $\chi\delta r$ and the direct derivative $\delta q^{\mathrm d}$ are uniquely
  determined on the active coordinates by the joint and own-block response
  systems, respectively.
  Their difference is
  \begin{equation}\label{eq:feedback}
    \chi\delta r-\delta q^{\mathrm d}
    =-J_{\mathcal U\mathcal U}^{-1}(J_{\mathrm o})_{\mathcal U\mathcal U}
    \delta q^{\mathrm d}.
  \end{equation}
\end{proposition}

Feedback transmits the direct adjustment through the rival blocks, allowing
the induced adjustments of all firms to change each firm's final response.
Appendix~\ref{app:response} proves the derivative identification.

Whereas opportunity changes alter the marginal contribution from capacity,
linear-cost changes directly alter its marginal cost.
Under strict complementarity and nearby nonnegative costs, a cost reduction
raises the marginal incentive to expand in the affected activities;
\eqref{eq:cost-response} gives its response system.
If only one firm's costs change, every other firm's frozen-rival
problem is unchanged.
Their live-rival responses then arise entirely through the recipient's
expansion and the adjustments it sets off.

The sign of feedback depends on relative appeal and activity overlap.
The rival blocks contain the difference between rival-plus-outside appeal
and own appeal, $B_{i\ell}-A_{i\ell}$, so they admit both signs
(Appendix~\ref{app:response-jacobian}).
When own appeal is smaller than rival-plus-outside appeal, rival expansion
lowers the marginal incentive to expand at that node; appeal dominance
reverses this effect, and overlapping activities combine the
node-specific effects.
Whether rivalry attenuates, amplifies or redirects a response is therefore a
property of the economy's appeal and overlap, rather than a general sign restriction.

\subsection{Scale and composition}\label{sec:model-composition}

A firm can expand without changing its business mix.
Raising capacities in two activities from $(40,60)$ to $(80,120)$ doubles
scale while leaving composition at $(0.4,0.6)$.
Raising them instead to $(60,60)$ changes the shares to $(0.5,0.5)$, although
neither activity contracts: a change in normalized mix is distinct from a
physical transfer of installed capacity.
More generally, a capacity response $h$ changes total scale by
$\mathbf1^\top h$, and at positive total capacity, differentiating normalized
composition gives
\begin{equation}\label{eq:composition-derivative}
  \frac{d}{dt}\frac{q_i+th}{\mathbf1^\top(q_i+th)}\bigg|_{t=0}
  =\frac{h-p_i\mathbf1^\top h}{S_i}.
\end{equation}
The subtraction removes proportional expansion: for $h=\alpha q_i$, the
numerator is $\alpha q_i-\alpha p_i S_i=0$.
Thus a large capacity response need not imply a large change in business mix.

For a vector $h$, write $D p_{q_i}(h)$ for the right-hand side of
\eqref{eq:composition-derivative}.
Applied to full minus direct adjustment, this derivative removes precisely
the feedback proportional to initial capacity.
An activity-separable equilibrium illustrates the remaining business-mix
response.
Under strict substitution, a marginal cost reduction in one activity expands
the recipient
there and crowds rivals out of that activity.
Positive capacity elsewhere then turns those changes into business-mix movement.
The activity-local cost shock considered here differs from the relative
opportunity news used in Section~\ref{sec:simulation}.
The result establishes business-mix movement, while movement toward or away
from competitors requires additional exposure ordering
(Appendix~\ref{app:two-site-example}).

\begin{proposition}[Proportional feedback and business-mix
  response]\label{prop:composition-feedback}
  Let $S_i>0$ and let $h_i^{\mathrm f}$ and $h_i^{\mathrm d}$ be full and
  direct capacity directions at the same initial allocation.
  Then
  \begin{equation}\label{eq:proportional-feedback}
    Dp_{q_i}(h_i^{\mathrm f})=Dp_{q_i}(h_i^{\mathrm d})
    \quad\Longleftrightarrow\quad
    h_i^{\mathrm f}-h_i^{\mathrm d}=\lambda_i q_i
    \text{ for some }\lambda_i\in\mathbb R.
  \end{equation}
  For an equilibrium interpretation, maintain an interior differentiable path
  and a fixed active set.
  In an activity-separable economy with identity demand loading and positive
  diagonal cost curvature, consider a marginal linear-cost reduction $s>0$ at
  one activity of firm $j$.
  At that activity let $A_i>0$, $D=b+\sum_k A_k$, $B_i=D-A_i$, $r,b>0$, and
  assume strict substitution $A_i<B_i$ for every firm.
  The recipient expands there and every unshocked rival contracts.
  Other activities remain unchanged.
  Each firm with positive capacity at another activity therefore
  changes business mix.
\end{proposition}

\begin{proof}[Proof idea]
  By \eqref{eq:composition-derivative}, zero normalized response means
  $h=p_i\mathbf1^\top h$, which is precisely a multiple of $q_i$.
  Apply this identity to the difference of the full and direct directions.
  For the separable benchmark, the activity Jacobian is diagonal plus rank
  one; its explicit solution has a positive recipient coordinate and negative
  rival coordinates.
  Unaffected activities solve homogeneous invertible systems.
  Appendix~\ref{app:composition-feedback} gives that solution and the
  composition argument.
\end{proof}

Attenuating a nonzero composition response changes the firm's
position even if it leaves the response's direction unchanged.

\subsection{Cash flow and competitive news}\label{sec:model-cash-flow}

The activity portfolio creates an exposure before any firm adjusts.
Write $X_{i\ell}=A_{i\ell}/D_\ell$ for firm $i$'s opportunity share, with
$D_\ell=b_\ell+\sum_j A_{j\ell}$; at fixed capacities, an opportunity change
$dr$ changes cash flow by $\sum_\ell X_{i\ell}\,dr_\ell$, so shared
opportunity exposure can generate correlated cash-flow news even before
competitive adjustment (Corollary~\ref{cor:exposure}).

The cash-flow envelope measures how much of an improvement in market
opportunities becomes additional operating cash flow once competitors respond.
Under strict complementarity, let $u=\chi\delta r$ be the equilibrium
capacity response to an opportunity direction $\delta r$, holding cost
primitives fixed, with block $u_j$ for firm $j$.
Write $k_\ell^\top$ for row $\ell$ of the activity-to-demand map $K$.
The envelope then gives
\begin{equation}\label{eq:equilibrium-cash-flow}
  \begin{split}
    \frac{d}{dt}\pi_i\bigl(q^*(r+t\delta r);r+t\delta r\bigr)\Big|_{t=0}
    ={}&\underbrace{\sum_\ell X_{i\ell}\,\delta r_\ell}
    _{\mathclap{\text{value of improved opportunities}}}\\
    &-\underbrace{\sum_\ell\frac{r_\ell A_{i\ell}}{D_\ell^2}
    \sum_{j\ne i}k_\ell^\top u_j}
    _{\mathclap{\text{effect of competitors' capacity changes}}}.
  \end{split}
\end{equation}
The first term values the disturbance at the firm's existing opportunity
shares.
The second evaluates how competitors' responses alter the contribution margin
available to the firm: $\sum_{j\ne i}k_\ell^\top u_j$ is their combined change
in appeal at node $\ell$, so their expansion lowers the value of the improved
opportunity and their contraction raises it.
At the optimizing allocation the firm's own marginal payoff is zero on
coordinates where capacity moves, so its own adjustment adds no first-order
cash-flow term.

Remark~\ref{rem:envelope} derives the split.
The payoff is stylized unlevered operating cash flow, abstracting from taxes,
working capital and depreciation; the envelope therefore supplies the operating
cash-flow component of a broader valuation, to which financing and discounting
would be added.

Holding a firm's own appeal fixed, competitors' expansion at a node lowers its
share there in proportion to the resulting growth in total appeal, so the
firm's loss depends on where it earns its margin and on how strongly a change
in margin moves its cash flow.
To express this dependence, let the share of firm $i$'s contribution margin
earned at node $\ell$ be
$\psi_{i\ell}=r_\ell X_{i\ell}/\mathcal M_i$.
Its operating leverage $\Lambda_i=\mathcal M_i/\pi_i$, contribution margin per
unit of operating cash flow, scales a change in margin into a cash-flow return.
For a capacity direction $u$, the proportional growth of total appeal at node
$\ell$ is $g_\ell=\sum_j k_\ell^\top u_j/D_\ell$, which is the same
for every firm.

\begin{proposition}[Common and differential competitive
  news]\label{prop:common-news}
  Maintain Assumption~\ref{ass:primitives} at an equilibrium with strict
  complementarity, and let $u$ be a capacity direction supported on its
  positive coordinates, such as the response to an opportunity or a
  linear-cost disturbance.
  For a firm with $\pi_i>0$, the rival term of the cash-flow change along $u$,
  expressed as a cash-flow return, is
  \begin{equation}\label{eq:rival-return}
    -\Lambda_i\sum_\ell\psi_{i\ell}\Bigl(g_\ell-\frac{k_\ell^\top
    u_i}{D_\ell}\Bigr).
  \end{equation}
  Let $\vartheta_i=\Lambda_i\psi_i$ and let
  $\zeta_i=-\vartheta_i^\top g$ denote the common-node component, before the
  own-appeal correction.
  Firms with equal operating leverage and contribution mix receive identical
  common-node news.
  For any square-integrable node-growth distribution with positive total
  common-node second moment, the retained share after removing the
  cross-sectional mean is
  $\varrho=\mathbb E\sum_i(\zeta_i-\bar\zeta)^2/
  \mathbb E\sum_i\zeta_i^2$, where $\bar\zeta=N^{-1}\sum_i\zeta_i$.
  Equal contribution mixes and leverage give $\varrho=0$.
\end{proposition}

Because firms face the same node-level growth in total appeal, differences
in their contribution mixes and operating leverage determine differences in
common-node news.
A comparison of firms within a period removes the mean of this
component $\zeta_i$, leaving the dispersion of $\vartheta_i$ in the directions
in which total appeal moves.
The own-appeal correction in~\eqref{eq:rival-return} removes each firm's own
contribution to total appeal growth, completing the rival effect; its
cross-firm variation must be retained separately.
The retained share $\varrho$ is a fraction of squared common-node response across firms,
not a fraction of the aggregate cash-flow loss.
The size of the competitive effect and the cross-sectional variation it
leaves are thus distinct: competitors' expansion can lower every firm's cash
flow substantially while barely changing how firms' cash-flow returns differ.
Appendix~\ref{app:common-news} gives the second-moment expansion and proof.

\section{Theory and numerical support}\label{app:theory-support}

The equilibrium proof combines the curvature of proportional shares with
strictly curved capacity costs; the response proof then restricts that field
to the locally unchanged active coordinates.
The composition and cash-flow arguments identify which parts of the resulting
capacity response survive normalization and comparisons across firms.
Numerical certificates bound operating-solver error before article sampling.

\subsection{Contest inheritance and the capacity
adaptation}\label{app:contest-inheritance}

The economic technology in \eqref{eq:firm-problem} adapts the multi-contest
framework of \citet{xu_zenou_zhou_conflict_2022}.
Their Equations~1--2 allocate effort directly to contests, while Remark~5,
Equation~54 and its footnote allow a positive, contest-specific denominator
term.
Setting contest production to the identity and $K=I$ recovers the corresponding
regularized share rule.
Here $Kq_i$ instead maps operating capacities to appeal across demands, and
outside appeal remains positive rather than being taken to zero in an
existence approximation.

Their Proposition~1 characterizes Nash equilibrium by a variational
inequality; Propositions~2--3 and Theorem~2 use monotonicity for uniqueness.
Theorem~3 derives local active-coordinate responses in their pure-cost setting
under their cost assumptions and nondegeneracy,
which corresponds to the strict complementarity condition used here.

The capacity adaptation requires its own hypotheses.
Xu and colleagues maintain increasing convex costs and impose strong increase
where needed; Assumption~\ref{ass:primitives} instead uses uniformly positive
quadratic curvature, including a general symmetric cost matrix.
Positive definiteness alone does not imply coordinatewise increasing cost.
Although one may represent feasible appeals as
$\mathcal A_i=\{Kq_i:q_i\geq0\}$ and minimize capacity cost over each appeal's
preimage, the resulting effective cost need not inherit their smoothness and
monotonicity assumptions.
Their restricted-contest Proposition~5 therefore does not make the present
orthant result an automatic corollary.
The following proofs establish the required concavity, monotonicity,
existence and sensitivity for this specification.

\subsection{Anchored operating stages}\label{app:anchored-stage}

For a fixed preceding allocation $a$, the stage field is
$F(q;r)+\kappa(q-a)$.
Its difference between any two feasible profiles has inner product at least
$(\nu+\kappa)\|q-q'\|^2$ with $q-q'$, because the anchor cancels and the
static field is $\nu$-strongly monotone.
The same coercivity argument as in Appendix~\ref{app:equilibrium} gives
existence on the capacity orthant; strong monotonicity gives uniqueness.
The shifted linear cost $c_i-\kappa a_i$ need not be nonnegative, so this
argument uses the anchored-stage result rather than imposing the static
nonnegative-cost restriction on the shifted problem.
Every stage retains contemporaneous rival choices and a fixed inherited
anchor, without a continuation-value first-order condition.

\subsection{Discrete-choice interpretation of proportional
appeal}\label{app:logit}

The interpretation in Remark~\ref{rem:logit} follows the conditional-logit
random-utility construction of
\citet{mcfadden_conditional_logit_1974} and
\citet{train_discrete_choice_2009}.
Fix a demand node $\ell$ and nonnegative integer counts $q_{ia}$.
Consumer choice is among all individual varieties and one outside option.
For variety $v$ belonging to firm $i$ at activity $a$, write deterministic
utility $V_v=-\tau_A\|x_a-z_\ell\|^2/d$; write $V_0=\log b_\ell$ for
outside utility.
The shocks $\epsilon_v$ are independent across alternatives with
$\Pr(\epsilon_v\leq e)=\exp\{-\exp(-e)\}$.
Independence and this distribution give
\begin{equation}\label{eq:logit-choice}
  \Pr(v\text{ chosen}\mid z_\ell)
  =\frac{\exp V_v}{\sum_w\exp V_w}.
\end{equation}
Summing this probability over firm $i$'s varieties yields
\begin{equation}\label{eq:logit-firm-share}
  \Pr(i\mid z_\ell)
  =\frac{\sum_a q_{ia}\exp\{-\tau_A\|x_a-z_\ell\|^2/d\}}
  {b_\ell+\sum_j\sum_a q_{ja}\exp\{-\tau_A\|x_a-z_\ell\|^2/d\}}
  =\frac{A_{i\ell}}{A_{i\ell}+B_{i\ell}}.
\end{equation}
For real nonnegative $q$, the same ratio defines a weighted-choice
extension: an active firm-activity alternative has deterministic utility
$\log q_{ia}+\log K_{\ell a}$, and zero-weight alternatives are absent.
This extension retains the share formula without interpreting a fractional
mass as a fractional number of independent shocks.
Multiproduct aggregation is also central to
\citet{nocke_schutz_aggregative_2018}, whose endogenous-price model
has different strategic choices from this fixed-margin capacity game.

\begin{remark}[Logit interpretation]\label{rem:logit}
  The share rule is multinomial logit \citep{mcfadden_conditional_logit_1974}
  when $q_{ia}$ measures varieties, outlets, or participation and $c_{ia}$
  is linear cost per unit.
  The parameter $\tau_A$ measures mismatch cost relative to taste heterogeneity,
  and $r_\ell=M_\ell m_\ell$ combines market size with fixed per-sale margin.
  Conditional on other primitives, size and margin changes are
  observationally equivalent in firm cash flows.
\end{remark}

\subsection{Mathematical preliminaries and notation}\label{app:proofs}

All spaces in the economic model are finite dimensional.
Let $I=\{1,\ldots,N\}$, $\mathcal A=\{1,\ldots,m\}$, and
$\mathcal L=\{1,\ldots,L\}$ index firms, activities, and demand
nodes, with $N,m,L\geq1$.
Firm $i$ chooses $q_i\in\mathcal Q_i=\mathbb R_+^m$ and the joint
profile belongs to
$\mathcal Q=\prod_i\mathcal Q_i$.
We use the ordinary Euclidean product structure:
\begin{equation}\label{eq:app-inner}
  \langle q,v\rangle=\sum_i\sum_a q_{ia}v_{ia},\qquad
  \|q\|^2=\sum_i\|q_i\|^2.
\end{equation}
The cone $\mathcal Q$ is convex, so segments and one-sided deviations remain
feasible even at a boundary allocation.

The appeal matrix $K\in\mathbb R_+^{L\times m}$ is common to firms.
Its row $k_\ell^\top$ maps a capacity block into the scalar appeal
$A_{i\ell}=k_\ell^\top q_i$.
Put $B_{i\ell}=b_\ell+\sum_{j\ne i}A_{j\ell}$ and
$D_\ell=b_\ell+\sum_i A_{i\ell}=A_{i\ell}+B_{i\ell}$.
Outside appeal $b_\ell>0$ guarantees $D_\ell\geq b_\ell$ at every
feasible profile.
Opportunity weights $r_\ell\geq0$ are in cash-flow units, and
$R_\Sigma=\sum_\ell r_\ell$ is the total available opportunity.
No rank condition on $K$ is required: different capacity changes can
have the same appeal image.

Costs are $C_i(q_i)=c_i^\top q_i+q_i^\top H_i q_i/2$, with $c_i\geq0$,
$H_i=H_i^\top$, and $H_i\succeq\nu I$ for a common $\nu>0$.
The Gaussian benchmark economy uses the same $H$ across firms and
heterogeneous $c_i$.
The proofs permit different $H_i$ when the same lower bound holds.
The negative marginal-profit field is $F_i=-\nabla_i\pi_i$; its
components have units of
cash flow per unit capacity.
Cost curvature has units of cash flow per squared unit capacity.
Consequently the error and regret bounds in
Appendix~\ref{app:certificates} have capacity
and payoff units, respectively.

\subsection{Proportional-share payoff algebra}\label{app:payoff}

The margin component of the operating payoff and the individual capacity
choice are, in full notation,
\begin{equation}\label{eq:revenue}
  \mathcal M_i(q_i;r_t,B_i)=\sum_\ell r_{\ell t}
  \frac{A_{i\ell}(q_i)}{A_{i\ell}(q_i)+B_{i\ell}}.
\end{equation}
\begin{equation}\label{eq:firm-choice}
  q_i^*(t)=\arg\max_{q_i\geq0}\pi_i(q_i,q_{-i};r_t).
\end{equation}

For a unilateral deviation, fix rivals and one demand node, writing current
own appeal as $a\geq0$, rival-plus-outside appeal as $B>0$ and $D=a+B$.
A feasible appeal change $s$ can have either sign, but $a+s\geq0$ keeps both
denominators strictly positive.
Subtracting the share margins gives the exact identity
\begin{align}
  \frac{r(a+s)}{D+s}-\frac{ra}{D}
  &=\frac{r[(a+s)D-a(D+s)]}{D(D+s)}
  =\frac{rBs}{D(D+s)}\label{eq:app-share}\\
  &=\frac{rBs}{D^2}-\frac{rBs^2}{D^2(D+s)}.\label{eq:share-expansion}
\end{align}
The nonnegative remainder measures the deduction from valuing the deviation
at the initial marginal contribution, for expansion and contraction alike.

For a feasible replacement $z_i$, let $d_i=z_i-q_i$ and
$s_{i\ell}=k_\ell^\top d_i$.
Linearity of appeal gives $A_{i\ell}(z_i)=A_{i\ell}(q_i)+s_{i\ell}$.
Applying \eqref{eq:share-expansion} at every node yields
\begin{align}
  \mathcal M_i(z_i;r,B_i)-\mathcal M_i(q_i;r,B_i)
  &=\sum_\ell\frac{r_\ell B_{i\ell}}{D_\ell^2}k_\ell^\top d_i
  -\mathcal R_i^{\mathrm{sat}}(q,d_i),\label{eq:app-revenue-expansion}\\
  \mathcal R_i^{\mathrm{sat}}(q,d_i)
  &=\sum_\ell\frac{r_\ell B_{i\ell}(k_\ell^\top d_i)^2}
  {D_\ell^2(D_\ell+k_\ell^\top d_i)}\geq0.\label{eq:app-remainder}
\end{align}
Symmetry of $H_i$ gives the exact cost difference
\begin{equation}\label{eq:app-cost-expansion}
  C_i(q_i+d_i)-C_i(q_i)
  =(c_i+H_i q_i)^\top d_i+\tfrac12d_i^\top H_i d_i.
\end{equation}
Subtracting costs from margin gives, with the field in \eqref{eq:field},
\begin{equation}\label{eq:payoff-expansion}
  \pi_i(z_i,q_{-i})-\pi_i(q)
  =-F_i(q)^\top d_i-\tfrac12d_i^\top H_i d_i-\mathcal R_i^{\mathrm{sat}}(q,d_i).
\end{equation}
Thus a nonpositive first-order gain rules out a profitable deviation of any size.

For later use, consider a fraction $t\in[0,1]$ of a feasible deviation.
The new denominator is at least $b_\ell$, so
\begin{equation}\label{eq:app-remainder-order}
  0\leq\mathcal R_i^{\mathrm{sat}}(q,td_i)
  \leq t^2\sum_\ell\frac{r_\ell B_{i\ell}(k_\ell^\top d_i)^2}{D_\ell^2 b_\ell}.
\end{equation}
The right-hand coefficient is finite and independent of $t$.
This explicit bound justifies taking a directional limit at boundary capacities,
including profiles where some activities remain inactive.

\subsection{Equilibrium, uniqueness, and symmetry}\label{app:equilibrium}

The marginal incentive to expand is summarized by a vector field $F:\mathcal
Q\to\mathbb R^{Nm}$ whose $i$th block is the negative own-cash-flow gradient,
so that its coordinates are marginal expansion cost less marginal contribution:
\begin{equation}\label{eq:field}
  F_{ia}(q)=c_{ia}+(H_i q_i)_a-
  \sum_\ell r_\ell\frac{B_{i\ell}(q_{-i})}
  {[A_{i\ell}(q_i)+B_{i\ell}(q_{-i})]^2}K_{\ell a}.
\end{equation}
A negative coordinate is an incentive to add the corresponding capacity, and
at zero capacity a positive coordinate makes entry unprofitable.
At an equilibrium every activity with positive capacity has a zero
coordinate: its marginal contribution equals its marginal capacity cost.

Proposition~\ref{prop:equilibrium} follows by identifying Nash equilibria with
solutions of the variational inequality
\begin{equation}\label{eq:vi}
  \langle F(q^*),q-q^*\rangle\geq0\qquad\text{for every }q\in\mathcal Q,
\end{equation}
then establishing existence and strong monotonicity on the entire orthant.

\subsubsection{Own concavity and the Nash--VI equivalence}

Holding rivals fixed, the own margin Hessian is
\begin{equation}\label{eq:app-hessian}
  \nabla_i^2\mathcal M_i(q_i;r,B_i)
  =-2\sum_\ell\frac{r_\ell B_{i\ell}}{(A_{i\ell}+B_{i\ell})^3}
  k_\ell k_\ell^\top.
\end{equation}
Its quadratic form is nonpositive, and subtracting $H_i$ makes the
payoff Hessian
bounded above by $-\nu I$.
The formula is valid on an open neighborhood of each feasible profile
because all denominators
are strictly positive there.
Thus each own payoff is strictly concave on its convex feasible set,
with a unique
best response whenever a maximum exists.

Suppose $q^*$ solves the joint VI \eqref{eq:vi}.
Choose a feasible joint alternative that changes only firm $i$.
The VI then says $F_i(q^*)^\top(z_i-q_i^*)\geq0$.
The exact payoff expansion, positive cost curvature, and nonnegative
remainder imply
$\pi_i(z_i,q_{-i}^*)\leq\pi_i(q^*)$ for every feasible $z_i$.
This holds for every firm, so $q^*$ is Nash.
For the converse, suppose $q$ is Nash and fix $z_i\geq0$ with $d_i=z_i-q_i$.
The intermediate choice $q_i+td_i$ is feasible for $0<t\leq1$, and
Nash optimality gives
\begin{equation}\label{eq:app-nash-direction}
  0\geq-tF_i(q)^\top d_i-\tfrac12t^2d_i^\top H_i d_i-\mathcal
  R_i^{\mathrm{sat}}(q,td_i).
\end{equation}
Divide by $t$ and use \eqref{eq:app-remainder-order} as $t\downarrow0$.
It follows that $F_i(q)^\top d_i\geq0$.
Sum these inequalities across the blocks of an arbitrary feasible
joint alternative
to obtain \eqref{eq:vi}.
This proves necessity and sufficiency, including at the boundary of the orthant.

\subsubsection{Demand-side monotonicity}

Own concavity alone does not establish joint monotonicity: rival
responses also enter the joint field.
We therefore calculate its demand part explicitly.
At one node let $a=(a_1,\ldots,a_N)$ denote the nonnegative appeal vector and
$D=b+\sum_i a_i$.
The negative marginal-contribution field in appeal coordinates is
$f_i(a)=-r(D-a_i)/D^2$.
For a direction $u\in\mathbb R^N$ and $U=\sum_i u_i$, direct
differentiation gives
\begin{equation}\label{eq:app-node-derivative}
  (Df(a)u)_i=\frac r{D^3}\bigl[D u_i+(D-2a_i)U\bigr].
\end{equation}
The raw Jacobian need not be symmetric.
Its quadratic form, which equals that of its symmetric part, is
\begin{align}
  u^\top Df(a)u
  &=\frac r{D^3}\left[D\left(\sum_i u_i^2+U^2\right)-2U\sum_i a_i
  u_i\right]\label{eq:node-monotonicity}\\
  &=\frac r{D^3}\left[
    b\left(\sum_i u_i^2+U^2\right)
  +\sum_i a_i\left\{\sum_{j\ne i}u_j^2+(U-u_i)^2\right\}\right]\geq0.
  \label{eq:app-demand-squares}
\end{align}
The second equality expands $D=b+\sum_i a_i$ and completes a square
separately for each $a_i$.
All coefficients are nonnegative.
This is the property of the proportional-share technology that preserves
monotonicity in the interacting game.

For capacity profiles $q,v$, put $h=q-v$ and $q(t)=v+th$.
At node $\ell$, the corresponding appeal direction is
$u_{i\ell}=k_\ell^\top h_i$.
The margin part $F^R$ of the capacity field is the pullback of the
node fields through $K$.
The fundamental theorem of calculus therefore yields
\begin{equation}\label{eq:app-integrated-demand}
  \langle F^R(q)-F^R(v),h\rangle
  =\int_0^1\sum_\ell u_\ell^\top Df_\ell(A_\ell(q(t)))u_\ell\,\mathrm dt
  =\mathcal M(q,v)\geq0.
\end{equation}
The finite sums and continuous derivatives justify the integral.
A nonzero $h_i$ may lie in the null space of $K$, so this argument
need not supply
strict curvature in all capacity directions.
That strictness comes from the resource-cost matrices.
The node inequality is the unit-exponent case of the marginal-share Jacobian
curvature in \citet[Lemma~A.4]{ewerhart_quartieri_contests_2020}, with the
outside coordinate fixed and its directional change set to zero.
Equation~\eqref{eq:app-integrated-demand} transfers that curvature to capacity;
the cost bound controls directions in the null space of $K$, supplying the
strong monotonicity required by the general variational-inequality result
\citep{parise_ozdaglar_network_games_2019}.

\subsubsection{Gaussian cost curvature}

The Gaussian congestion kernel $K^c$ is positive semidefinite on any finite
set.
When $\tau_c=0$, it is the all-ones matrix and
$v^\top K^c v=(\sum_a v_a)^2\geq0$.
When $\tau_c>0$, factor its entries and expand the exponential:
\begin{equation}\label{eq:app-gaussian-series}
  K^c_{ab}=e^{-\tau_c\|x_a\|^2/d}e^{-\tau_c\|x_b\|^2/d}
  \sum_{n=0}^{\infty}\frac{(2\tau_c/d)^n}{n!}\langle x_a,x_b\rangle^n.
\end{equation}
For every nonnegative integer $n$, the matrix with entries
$\langle x_a,x_b\rangle^n$ is a Gram matrix of the tensor vectors
$x_a^{\otimes n}$;
for $n=0$ use the scalar vector one.
Consequently
\begin{equation}\label{eq:app-gaussian-psd}
  v^\top K^c v=\sum_{n=0}^{\infty}\frac{(2\tau_c/d)^n}{n!}
  \left\|\sum_a v_a e^{-\tau_c\|x_a\|^2/d}x_a^{\otimes n}\right\|^2\geq0.
\end{equation}
Finite activity sums and absolute convergence of the exponential
series justify the equality.
Thus $H=\eta K^c+\nu I$ has quadratic form at least $\nu\|v\|^2$ for $\eta\geq0$.
The same conclusion holds for any symmetric $H_i$ satisfying the stated bound.

\subsubsection{Existence on the unbounded cone}

The cost part of $F$ has difference $H_i(q_i-v_i)$ in block $i$.
Combining its curvature with \eqref{eq:app-integrated-demand} gives
\begin{equation}\label{eq:strong-monotonicity}
  \langle F(q)-F(v),q-v\rangle
  =\mathcal M(q,v)+\sum_i(q_i-v_i)^\top H_i(q_i-v_i)
  \geq\nu\|q-v\|^2.
\end{equation}
To establish existence, first obtain a strategy bound that holds
uniformly over rivals.
Each share lies between zero and one, hence
\begin{equation}\label{eq:app-capacity-bound}
  \pi_i(q_i,q_{-i})\leq R_\Sigma-c_i^\top q_i-\tfrac12q_i^\top H_i q_i
  \leq R_\Sigma-\tfrac\nu2\|q_i\|^2.
\end{equation}
Zero capacity earns zero for every rival profile.
Choose $R>\sqrt{2R_\Sigma/\nu}$ and let
$X_i=\{q_i\geq0:\|q_i\|\leq R\}$.
Every strategy with norm at least $R$ has negative payoff and is
dominated by zero.
The sets $X_i$ are nonempty, compact, and convex.
Continuity and strict own concavity give a unique maximizer over each
$X_i$ for every rival profile.

On the compact product $X=\prod_i X_i$, continuity of payoffs and uniqueness
of each maximizer make the best-response map continuous by the maximum theorem.
Brouwer's fixed-point theorem therefore supplies a fixed point in $X$.
No choice outside $X_i$ can improve on it, since those choices earn
negative payoff and zero is available inside.
The fixed point is therefore Nash on the entire orthant.

\subsubsection{Uniqueness and the symmetry boundary}\label{app:symmetry-proof}

Let $q^*$ and $v^*$ be two equilibria.
Their equivalent VI conditions, tested against one another, give
\begin{equation}\label{eq:app-uniqueness}
  \langle F(q^*),v^*-q^*\rangle\geq0,\qquad
  \langle F(v^*),q^*-v^*\rangle\geq0.
\end{equation}
Adding and reversing signs bounds the monotonicity pairing above by zero.
Its lower bound is $\nu\|q^*-v^*\|^2$, so the two equilibria coincide.
If primitives are invariant under relabeling firms, a permutation of
an equilibrium remains an equilibrium:
each permuted firm faces the same own cost and the correspondingly
permuted sum of rivals' appeals.
Uniqueness requires the profile to be invariant under every
transposition of firms.
Hence all firm blocks are equal.
Capability heterogeneity removes this permutation invariance and can
generate different equilibrium allocations.

\subsection{The no-interaction counterfactual}\label{app:independent-proof}

In the independent-firm economy, each firm faces only outside appeal while
retaining its own opportunity and cost primitives.
Its payoff is
\begin{equation}\label{eq:app-independent-payoff}
  \pi_i^{\mathrm M}(q_i)=\sum_\ell r_\ell\frac{k_\ell^\top
  q_i}{k_\ell^\top q_i+b_\ell}
  -c_i^\top q_i-\tfrac12q_i^\top H_i q_i.
\end{equation}
The associated field, derived by differentiation of this objective, is
\begin{equation}\label{eq:app-independent-field}
  F_i^{\mathrm M}(q_i)=c_i+H_i q_i-
  \sum_\ell\frac{r_\ell b_\ell}{(k_\ell^\top q_i+b_\ell)^2}k_\ell.
\end{equation}
With no rival blocks in its derivative, this operator is a Cartesian
product of one-firm proportional-share problems.
Each firm retains access to the full $R_\Sigma$, so aggregate market margin
can differ between the independent and strategic economies.

Each block Jacobian is
\begin{equation}\label{eq:app-independent-curvature}
  DF_i^{\mathrm M}(q_i)=H_i+
  2\sum_\ell\frac{r_\ell b_\ell}{(k_\ell^\top q_i+b_\ell)^3}k_\ell k_\ell^\top
  \succeq\nu I.
\end{equation}
Integration along each block segment and summation over firms gives
$\langle F^{\mathrm M}(q)-F^{\mathrm M}(v),q-v\rangle\geq\nu\|q-v\|^2$.
The payoff bound \eqref{eq:app-capacity-bound} still applies.
A maximum therefore exists for each block and strict concavity makes it unique;
the vector of these maxima is the unique independent equilibrium.
The exact payoff expansion remains valid with $B_{i\ell}=b_\ell$.
The VI equivalence and the distance and regret arguments below follow
directly for this product operator.
Under identical capabilities every independent optimum is identical,
even though its level need not equal
the symmetric strategic equilibrium.

\subsection{Response to opportunity and cost shifts}\label{app:response}

Throughout, $q$ is feasible, so $D_\ell\geq B_{i\ell}\geq b_\ell>0$ for every
firm and node.
Write $F(q;r)$ for the field \eqref{eq:field} at opportunity values $r$.

\subsubsection{Opportunity continuity}

Under Assumption~\ref{ass:primitives}, equilibria at nonnegative opportunity
vectors $r,r'$ satisfy
\begin{equation}\label{eq:response-continuity}
  \|q^*(r')-q^*(r)\|\leq\frac{L_r}{\nu}\|r'-r\|,
  \qquad L_r^2=N\sum_\ell\frac{\|k_\ell\|^2}{b_\ell^2}.
\end{equation}
This bound requires neither a fixed active set nor differentiability of the
equilibrium map.

Only the opportunity weights change, so the cost terms cancel and
\begin{equation}\label{eq:app-field-shift}
  F_{ia}(q;r')-F_{ia}(q;r)=-\sum_\ell(r'_\ell-r_\ell)\mathcal G_{ia,\ell},
  \qquad
  0\leq\mathcal G_{ia,\ell}=\frac{B_{i\ell}}{D_\ell^2}K_{\ell
  a}\leq\frac{K_{\ell a}}{b_\ell}.
\end{equation}
The upper bound uses $B_{i\ell}\leq D_\ell$ and $b_\ell\leq D_\ell$.
By Cauchy--Schwarz, each coordinate of the shift has square at most
$\|r'-r\|^2\sum_\ell K_{\ell a}^2/b_\ell^2$.
Summing over firms and activities gives
$\|F(q;r')-F(q;r)\|\leq L_r\|r'-r\|$ at every feasible $q$.
Let $q=q^*(r)$ and $q'=q^*(r')$.
Adding the two variational inequalities \eqref{eq:vi}, each evaluated
at the other
equilibrium, and using \eqref{eq:strong-monotonicity} at opportunity
values $r$ gives
\begin{equation}\label{eq:app-continuity}
  \nu\|q'-q\|^2\leq\langle F(q';r)-F(q';r'),q'-q\rangle
  \leq L_r\|r'-r\|\,\|q'-q\|,
\end{equation}
which is \eqref{eq:response-continuity}.
This is the strongly monotone sensitivity bound of
\citet{parise_ozdaglar_network_games_2019} with the model's loading bound.

\subsubsection{The Jacobian and its curvature}\label{app:response-jacobian}

Formally, write $k_\ell^\top$ for row $\ell$ of $K$ and $D_\ell=A_{i\ell}+B_{i\ell}$.
Raising $r_\ell$ lowers coordinate $(i,a)$ of the field \eqref{eq:field}
by exactly $\mathcal G_{ia,\ell}=B_{i\ell}K_{\ell a}/D_\ell^2$, so an
opportunity change $\delta r$ moves the field by $-\mathcal G\delta r$ at
fixed capacities.
The field's Jacobian $J$ has own and rival blocks
\begin{equation}\label{eq:jacobian}
  J_{ii}=H_i+\sum_\ell\frac{2r_\ell B_{i\ell}}{D_\ell^3}k_\ell k_\ell^\top,
  \qquad
  J_{ij}=\sum_\ell\frac{r_\ell(B_{i\ell}-A_{i\ell})}{D_\ell^3}k_\ell k_\ell^\top
  \quad(j\ne i).
\end{equation}
Let $\mathcal U$ be a fixed set of active firm-activity coordinates and $P$
the orthogonal projection onto them.
The projected linear response to an opportunity direction $\delta r$ is defined by
\begin{equation}\label{eq:linear-response}
  PJ\,\delta q=P\mathcal G\,\delta r,\qquad P\delta q=\delta q.
\end{equation}
Its $i$th block is~\eqref{eq:firm-response} on the active coordinates.
The direct response $\delta q^{\mathrm d}$ solves this system with $J$
replaced by $J_{\mathrm d}=\operatorname{blockdiag}(J_{ii})$.
It retains the strategic origin's rival fields $B_i$, rather than replacing
them by outside appeal alone.
The difference between the responses in Proposition~\ref{prop:response}
is the direct adjustment transmitted through the rival blocks
$J_{\mathrm o}=J-J_{\mathrm d}$.

Along a line $q+tu$ through a feasible profile, node $\ell$ contributes
$r_\ell(D_\ell-A_{i\ell})D_\ell^{-2}k_\ell$ to firm $i$'s block, with
numerator and
denominator affine in $t$.
Write $U_\ell=\sum_{j}k_\ell^\top u_j$ for the total appeal change and
$u_{i\ell}=k_\ell^\top u_i$ for firm $i$'s own change.
The quotient rule gives the derivative
\begin{equation}\label{eq:app-node-quotient}
  r_\ell\Bigl[\frac{U_\ell-u_{i\ell}}{D_\ell^2}
  -\frac{2(D_\ell-A_{i\ell})U_\ell}{D_\ell^3}\Bigr]k_\ell
  =-\frac{r_\ell}{D_\ell^3}\bigl[2B_{i\ell}u_{i\ell}
  +(B_{i\ell}-A_{i\ell})(U_\ell-u_{i\ell})\bigr]k_\ell,
\end{equation}
so the negative marginal-profit field has directional derivative $Ju$ with the
blocks in \eqref{eq:jacobian}.
The pairing $\langle Ju,u\rangle$ is therefore the derivative of
$t\mapsto\langle F(q+tu),u\rangle$ at zero.
That derivative is $\sum_{i}u_i^\top H_i u_i$ plus the node form
\eqref{eq:node-monotonicity}, which is nonnegative, so
$\langle Ju,u\rangle\geq\nu\|u\|^2$.
The own-block part contributes
$\sum_{i}u_i^\top H_i u_i+\sum_{i,\ell}2r_\ell B_{i\ell}D_\ell^{-3}u_{i\ell}^2$,
so $J_{\mathrm d}$ has the same modulus.
Neither bound uses an equilibrium or an active set.

\subsubsection{Existence, uniqueness, and size of the response}

Let $T$ be either $J$ or $J_{\mathrm d}$, and suppose $u$ solves
$PTu=Pg$ with $Pu=u$.
Because $P$ is an orthogonal projection,
\begin{equation}\label{eq:app-restricted-bound}
  \nu\|u\|^2\leq\langle Tu,u\rangle=\langle Tu,Pu\rangle
  =\langle PTu,u\rangle=\langle Pg,u\rangle\leq\|Pg\|\,\|u\|.
\end{equation}
Hence $\|u\|\leq\|Pg\|/\nu$.
Differences of solutions solve the system with zero forcing, so the
solution is unique.
The linear map $u\mapsto PTPu+(u-Pu)$ has trivial kernel by the same bound.
It is therefore onto in finite dimension, and a preimage of $Pg$ solves the
restricted system.
Thus at every feasible profile and for any fixed active set, both the full
and direct projected systems have unique solutions bounded by
$\|P\mathcal G\delta r\|/\nu$.
At an equilibrium with strict complementarity the derivative identification
below gives the economic interpretation in Proposition~\ref{prop:response}.

\subsubsection{Direct response and rival feedback}

Subtracting the direct system from the full system gives
$PJ(\chi\delta r-\delta q^{\mathrm d})=-PJ_{\mathrm o}\delta q^{\mathrm d}$.
Both directions lie in the active coordinate space, where $J$ is invertible
by \eqref{eq:app-restricted-bound}.
Applying that inverse gives \eqref{eq:feedback}, which proves the
common-origin comparison in Proposition~\ref{prop:response}.

\subsubsection{The derivative of the equilibrium
map}\label{app:equilibrium-derivative}

Suppose strict complementarity holds at $q=q^*(r)$: each coordinate has positive
capacity or a strictly positive field value.
With $\mathcal U$ the coordinates of positive capacity,
\begin{equation}\label{eq:equilibrium-derivative}
  \|q^*(r')-q^*(r)-\chi\,(r'-r)\|=o(\|r'-r\|),
\end{equation}
so $\chi$ is the derivative of the equilibrium map and inactive
coordinates do not
respond.

The field is a rational function of capacities with positive
denominators, so it is
Fréchet differentiable, and its derivative along every line is $J$
(Appendix~\ref{app:response-jacobian}).
At an equilibrium every field coordinate is nonnegative and vanishes
where capacity
is positive.
Let $m>0$ bound from below each positive capacity and
each positive field value at $q$.
By \eqref{eq:response-continuity} and the Lipschitz continuity of $F$
on the orthant,
for $\|r'-r\|$ small the equilibrium $q'=q^*(r')$ keeps every
positive coordinate
positive, and every coordinate with a positive field value stays at zero.
Both equilibria therefore satisfy $PF=0$ with $P(q'-q)=q'-q$.
Expanding, $F(q';r')-F(q;r)=J(q'-q)-\mathcal G\,\delta r+R$, where
$R$ collects the
Fréchet remainder and the change in $\mathcal G$ between $q$ and $q'$; both are
$o(\|q'-q\|+\|\delta r\|)$.
The error $q'-q-\chi\delta r$ is supported on $\mathcal U$ and solves
the linearized
system with forcing $-R$, so \eqref{eq:app-restricted-bound} bounds it by
$\|PR\|/\nu=o(\|\delta r\|)$.
This proves \eqref{eq:equilibrium-derivative}.

Linear costs enter the field additively, $F(q;c')-F(q;c)=c'-c$.
The same argument with forcing $-(c'-c)$ in place of $\mathcal
G\,\delta r$ shows that,
under strict complementarity and for nearby nonnegative costs $c'$,
\begin{equation}\label{eq:cost-response}
  \|q^*(c')-q^*(c)-u\|=o(\|c'-c\|),\qquad PJu=-P(c'-c),\quad Pu=u.
\end{equation}
The forcing does not depend on capacities, so the remainder is the
Fréchet remainder
alone.
If $c'-c$ is supported on one firm, every other firm's problem
against frozen rivals
is unchanged, so its frozen-rival response is zero.
The opportunity and cost responses are established separately, each for one
type of disturbance at a time.

\subsection{Proof of the composition-feedback
benchmark}\label{app:composition-feedback}

For $S=\mathbf1^\top q>0$, the quotient rule gives
$Dp_q(h)=(h-p\mathbf1^\top h)/S$.
Its kernel consists exactly of $h=\lambda q$: substituting this form gives
zero, and a zero numerator gives $h=(\mathbf1^\top h/S)q$.
Linearity yields \eqref{eq:proportional-feedback}.

At the affected activity the negative marginal payoff is $g_i=k_i+\omega_i
A_i-rB_i/D^2$, where $k_i$ is its linear cost and $\omega_i>0$ its curvature.
Differentiating the interior first-order conditions gives
\begin{equation}\label{eq:separable-jacobian}
  J=\operatorname{diag}(a)+c\mathbf1^\top,\qquad
  a_i=\omega_i+r/D^2>0,\qquad c_i=r(B_i-A_i)/D^3>0.
\end{equation}
Indeed the own derivative is $\omega_i+2rB_i/D^3=a_i+c_i$ and every rival
derivative in row $i$ is $c_i$.
Let $u$ denote the capacity derivative and $T=\sum_i u_i$.
The cost reduction imposes $a_i u_i+c_i T=s\,\mathbf{1}_{\{i=j\}}$.
Dividing by $a_i$ and summing yields
\begin{equation}\label{eq:separable-response}
  \begin{aligned}
    \Delta&=1+\sum_i c_i/a_i>0,&
    T&=\frac{s/a_j}{\Delta}>0,\\
    u_j&=\frac{s}{a_j}\frac{1+\sum_{i\ne j}c_i/a_i}{\Delta}>0,&
    u_i&=-\frac{c_i}{a_i}T<0\quad(i\ne j).
  \end{aligned}
\end{equation}
For uniqueness, the difference of two solutions has aggregate $U$ satisfying
$\Delta U=0$, hence each difference coordinate is zero.
Each unaffected activity has zero forcing and an invertible block by
positive-curvature strong monotonicity, so its response vanishes.
If firm $i$ has positive capacity $q_{ib}$ at another activity $b$, its share
derivative there is $-q_{ib}u_i/S_i^2\ne0$.
This proves the business-mix claim under the maintained interior path.
Entry and exit require a directional analysis; these derivatives make no
assertion across active-set changes.

\subsection{Two-site exposure ordering}\label{app:two-site-example}

Whether a capacity response separates firms' business mixes depends on
their initial exposures.
To make that dependence explicit, consider two firms at two distinct activity
sites with identity demand loading and positive diagonal quadratic cost
curvature.
Initial interior capacities are $q_i=(q_{i1},q_{i2})$ and
$q_j=(q_{j1},q_{j2})$, all positive.
At activity 1, the node-1 opportunity value and outside appeal satisfy $r,b>0$,
and each firm's payoff component has the form
\begin{equation}\label{eq:sign-payoff}
  \pi_i^{(1)}(q_{i1},q_{j1})
  =\frac{rq_{i1}}{b+q_{i1}+q_{j1}}-c_i q_{i1}-\frac{\omega_i}{2}q_{i1}^2,
  \qquad \omega_i,\omega_j>0.
\end{equation}
Firm $j$ receives a small cost reduction $c_j\mapsto c_j-s$, preserving
nonnegative costs.
Assume a differentiable interior equilibrium path satisfying both first-order
conditions locally, with activity 2 unchanged.
To compare business mixes we use the quadratic Wasserstein distance $W_2$,
the least root-mean-square distance over which mass must move to turn one
composition into the other.
The two sites are a distance $d_*>0$ apart.
Define $P_k(s)$ from the normalized two-site capacities and
$Y^F_{ij}(s)=W_2^2(P_i(s),P_j(0))-W_2^2(P_i(0),P_j(0))$, the change in the
squared distance from firm $i$'s business mix to firm $j$'s initial mix.

\begin{proposition}[Conditional capacity and composition
  response]\label{prop:two-site-adjustment}
  In the two-site economy and along the path just specified, suppose
  \begin{equation}\label{eq:sign-conditions}
    q_{i1}<b+q_{j1},\qquad q_{j1}q_{i2}>q_{i1}q_{j2}.
  \end{equation}
  Then the cost recipient expands in activity 1, $\dot q_{j1}>0$, its competitor
  contracts, $\dot q_{i1}<0$, and $\dot Y^F_{ij}>0$.
  Reversing the strict initial exposure ordering reverses the sign of
  $\dot Y^F_{ij}$ while preserving the capacity-response signs.
\end{proposition}

\subsubsection{Proof of the two-site capacity and composition
response}\label{app:two-site-adjustment}

\begin{proof}[Proof of Proposition~\ref{prop:two-site-adjustment}]
  Differentiate the two interior first-order conditions in
  \eqref{eq:sign-payoff}.
  Write $D=b+q_{i1}+q_{j1}$, $B_i=b+q_{j1}$, $B_j=b+q_{i1}$, and define
  \begin{equation}\label{eq:sign-jacobian}
    \begin{aligned}
      J_{ii}&=\omega_i+2rB_i/D^3,& J_{jj}&=\omega_j+2rB_j/D^3,\\
      J_{ij}&=r(B_i-q_{i1})/D^3,& J_{ji}&=r(B_j-q_{j1})/D^3,\\
      \Delta&=J_{ii}J_{jj}-J_{ij}J_{ji}.
    \end{aligned}
  \end{equation}
  These are the own and rival entries of the field's
  Jacobian~\eqref{eq:jacobian} restricted to activity 1.
  The demand-only part of the determinant has numerator
  $r^2[4{B_i}B_j-(B_i-q_{i1})(B_j-q_{j1})]$ over $D^6$.
  The bracket equals $3b^2+4b(q_{i1}+q_{j1})+(q_{i1}+q_{j1})^2>0$.
  The additional terms from $\omega_i,\omega_j>0$ are positive, so $\Delta>0$.
  Solving the differentiated conditions gives
  \begin{equation}\label{eq:sign-capacity}
    \begin{pmatrix}J_{ii}&J_{ij}\\J_{ji}&J_{jj}
    \end{pmatrix}
    \begin{pmatrix}\dot q_{i1}\\\dot q_{j1}
    \end{pmatrix}
    =
    \begin{pmatrix}0\\1
    \end{pmatrix},
    \qquad \dot q_{i1}=-J_{ij}/\Delta,\quad \dot q_{j1}=J_{ii}/\Delta.
  \end{equation}
  The substitution inequality $q_{i1}<b+q_{j1}$ makes $J_{ij}>0$, yielding
  competitor contraction and recipient expansion.

  For shares $p,p'\in[0,1]$ on two sites separated by $d_*>0$, the exact
  squared transport distance is
  \begin{equation}\label{eq:two-site-transport}
    W_2^2\bigl((p,1-p),(p',1-p')\bigr)=d_*^2|p-p'|.
  \end{equation}
  Any feasible nonnegative coupling has off-diagonal masses whose
  difference is $p-p'$.
  Their sum is at least $|p-p'|$, and matching the common mass at each site
  attains the bound.
  Put $p_k=q_{k1}/(q_{k1}+q_{k2})$.
  Because activity 2 is unchanged,
  $\dot p_i=q_{i2}\dot q_{i1}/(q_{i1}+q_{i2})^2<0$.
  The exposure ordering implies $p_j(0)>p_i(0)$ and persists locally against
  the frozen $p_j(0)$.
  Differentiating the squared distance therefore gives
  \begin{equation}\label{eq:sign-distance}
    \dot Y^F_{ij}=-d_*^2\dot p_i
    =\frac{d_*^2{q_{i2}}J_{ij}}{(q_{i1}+q_{i2})^2\Delta}>0.
  \end{equation}
  If $p_j(0)<p_i(0)$, the derivative instead equals $d_*^2\dot p_i<0$.
  This establishes the statement and its reversed-order conclusion.
\end{proof}

Initial exposure ordering determines whether contraction moves the
competitor's business mix toward or away from the recipient's initial mix;
the capacity-response signs alone do not determine that direction.

\subsection{Cash-flow exposure and competitive news}\label{app:cash-flow-support}

\subsubsection{Cash-flow exposure}

The node shares also determine operating cash-flow exposure.
Write $X(q)$ for the matrix of node shares
$X_{i\ell}=A_{i\ell}/D_\ell$.

\begin{corollary}[Cash-flow exposure]\label{cor:exposure}
  At fixed capacities, payoff is affine in opportunity values,
  $\pi_i(q;r')-\pi_i(q;r)=\sum_\ell(r'_\ell-r_\ell)X_{i\ell}(q)$.
  If opportunity innovations $\varepsilon$ have covariance $\Omega$, then
  \begin{equation}\label{eq:exposure-covariance}
    \operatorname{Cov}\bigl(X(q)\varepsilon\bigr)=X(q)\,\Omega\,X(q)^\top.
  \end{equation}
  A rival's expansion weakly lowers the firm's payoff: if $q_j'\geq q_j$ for
  some $j\neq i$, then
  \begin{equation}\label{eq:rival-cross}
    \pi_i(q_j',q_{-j};r)\leq\pi_i(q;r).
  \end{equation}
\end{corollary}

Uncorrelated innovations therefore already produce correlated operating
cash flow whenever firms' appeal overlaps at common nodes.
Any strategic contribution to cash-flow covariance has to be measured
against this shared-exposure baseline.

\begin{proof}
  At fixed capacities, the margin is linear in $r$ with coefficients
  $A_{i\ell}/D_\ell$, and costs do not depend on $r$.
  The covariance formula \eqref{eq:exposure-covariance} is bilinearity of
  covariance for square-integrable innovations and a predetermined $X(q)$.
  If $q_j'\geq q_j$ for a rival $j$, then $A_{i\ell}$ and $C_i$ are unchanged
  while each $D_\ell$ weakly rises because $K\geq0$.
  Each term $r_\ell A_{i\ell}/D_\ell$ weakly falls, which gives
  \eqref{eq:rival-cross}.
\end{proof}

\begin{remark}[Envelope split]\label{rem:envelope}
  Along a direction $(u,\delta r)$ of capacities and opportunity
  values, firm $i$'s
  cash flow changes at the rate
  \begin{equation}\label{eq:envelope}
    \sum_\ell\delta r_\ell X_{i\ell}-\sum_a F_{ia}u_{ia}
    -\sum_\ell\frac{r_\ell A_{i\ell}}{D_\ell^2}\Bigl(\sum_{j}k_\ell^\top
    u_j-k_\ell^\top u_i\Bigr).
  \end{equation}
  At a static equilibrium $F_{ia}=0$ wherever the response moves the firm's
  capacity, so the change is direct opportunity news plus the value effect of
  rivals' movement.
  Along the equilibrium response $u=\chi\delta r$, under strict complementarity,
  the equilibrium cash-flow sensitivity is therefore
  given by \eqref{eq:equilibrium-cash-flow}:
  direct opportunity news plus the value of rivals' equilibrium repositioning.
\end{remark}

\begin{proof}
  The margin is linear in $r$, so the opportunity term is the share $X_{i\ell}$.
  The quotient rule gives each node's capacity derivative, and symmetry of $H_i$
  collects the own terms into $-F_i^\top u_i$.
  The remaining node terms involve only rivals' appeal changes.
  For \eqref{eq:equilibrium-cash-flow}, cash flow is Fréchet differentiable and
  shares are continuous, so \eqref{eq:equilibrium-derivative} replaces
  $q^*(r+t\delta r)-q^*(r)$ by $t\chi\delta r$ up to $o(t)$.
  The own term vanishes because $\chi\delta r$ is supported where $F$ vanishes.
\end{proof}

\subsubsection{Common and differential competitive news}\label{app:common-news}

The retained share in Proposition~\ref{prop:common-news} measures the
common-node component, separately from its own-appeal correction.
Let $\vartheta_i=\Lambda_i\psi_i$, with mean
$\bar\vartheta=N^{-1}\sum_i\vartheta_i$, and let $\zeta_i=-\vartheta_i^\top g$ be
the common-node part of~\eqref{eq:rival-return}, with cross-sectional mean
$\bar\zeta$.
Across disturbances, treat node growth $g$ as a random vector with
second-moment matrix $\Sigma_g=\mathbb E[gg^\top]$.
Whenever $\sum_i\vartheta_i^\top\Sigma_g\vartheta_i>0$, the fraction of the
common-node part's expected cross-sectional second moment that survives
removal of the cross-sectional mean is
\begin{equation}\label{eq:retained-share}
  \varrho=\frac{\mathbb E\sum_i(\zeta_i-\bar\zeta)^2}{\mathbb E\sum_i\zeta_i^2}
  =\frac{\sum_i(\vartheta_i-\bar\vartheta)^\top\Sigma_g(\vartheta_i-\bar\vartheta)}
  {\sum_i\vartheta_i^\top\Sigma_g\vartheta_i}.
\end{equation}
The common-node projection identity is algebraic; it does not establish
that demeaning by itself makes competitive news statistically undetectable.

\begin{proof}[Proof of Proposition~\ref{prop:common-news}]
  By \eqref{eq:envelope}, the cash-flow change along $(u,\delta r)$ is direct
  opportunity news, an own term $-\sum_a F_{ia}u_{ia}$, and the rival term
  $-\sum_\ell r_\ell A_{i\ell}D_\ell^{-2}(\sum_j k_\ell^\top
  u_j-k_\ell^\top u_i)$.
  The own term vanishes because $u$ is supported where $F$ vanishes.
  A linear-cost direction $\delta c$ adds the direct effect
  $-\delta c_i^\top q_i$, which does not involve rivals, so the rival
  term is unchanged.
  Since $r_\ell A_{i\ell}/D_\ell=r_\ell X_{i\ell}=\mathcal M_i\psi_{i\ell}$ and
  $\sum_j k_\ell^\top u_j-k_\ell^\top u_i=D_\ell g_\ell-k_\ell^\top u_i$,
  dividing the rival term by $\pi_i$ gives \eqref{eq:rival-return}.
  For the retained fraction,
  $\sum_i(\vartheta_i^\top g-\bar\vartheta^\top
  g)^2=\sum_i((\vartheta_i-\bar\vartheta)^\top g)^2$,
  whose expectation is
  $\sum_i(\vartheta_i-\bar\vartheta)^\top\Sigma_g(\vartheta_i-\bar\vartheta)$;
  the expected second moment before demeaning is
  $\sum_i\vartheta_i^\top\Sigma_g\vartheta_i$.
  If all $\vartheta_i$ are equal, every deviation is zero.
\end{proof}

\subsection{Numerical equilibrium and composition
certificates}\label{app:certificates}

Projected extragradient locates each equilibrium; independently evaluated
residuals determine whether to accept it.
For a feasible numerical profile $\widehat q$, define the orthant
stationarity residual by
\begin{equation}\label{eq:residual}
  s_{ia}(\widehat q)=
  \begin{cases}
    F_{ia}(\widehat q),&\widehat q_{ia}>0,\\
    \min\{F_{ia}(\widehat q),0\},&\widehat q_{ia}=0.
  \end{cases}
\end{equation}
At an inactive coordinate, positive $F_{ia}$ discourages entry and contributes
no error; a negative value records an unexploited incentive to enter.

\subsubsection{Capacity error and unilateral regret}\label{app:regret-proof}

\begin{proposition}[Equilibrium accuracy and unilateral payoff
  regret]\label{prop:certificate}
  Under Assumption~\ref{ass:primitives}, the strategic and
  independent operators each satisfy
  \begin{equation}\label{eq:certificate}
    \|\widehat q-q^*\|\leq\frac{\|s(\widehat q)\|}{\nu},\qquad
    \operatorname{Regret}_i(\widehat q)\leq\frac{\|s_i(\widehat q)\|^2}{2\nu},
  \end{equation}
  where regret is the supremum of the unilateral payoff improvement
  holding rivals fixed.
\end{proposition}

\begin{proof}
  Put $n=F(\widehat q)-s$.
  By construction $n\geq0$ and $n_{ia}\widehat q_{ia}=0$.
  For every feasible $z$,
  \begin{equation}\label{eq:app-approx-vi}
    \langle F(\widehat q),z-\widehat q\rangle
    =\langle s,z-\widehat q\rangle+\langle n,z\rangle
    \geq-\|s\|\|z-\widehat q\|.
  \end{equation}
  With $h=\widehat q-q^*$, this inequality at $z=q^*$ and the equilibrium VI
  give
  \begin{equation}\label{eq:app-distance-chain}
    \nu\|h\|^2\leq\langle F(\widehat q)-F(q^*),h\rangle
    \leq\|s\|\|h\|.
  \end{equation}
  Dividing by $\nu\|h\|$ when $h\ne0$ proves the distance bound; the case
  $h=0$ is immediate.
  \label{app:independent-distance}
  The independent product operator has the same modulus, so the bound applies
  to its unique equilibrium as well.

  For a feasible unilateral deviation $d_i=z_i-\widehat q_i$, the block form
  of \eqref{eq:app-approx-vi} and the exact payoff expansion give
  \begin{align}
    \pi_i(z_i,\widehat q_{-i})-\pi_i(\widehat q)
    &\leq\|s_i\|\|d_i\|-\tfrac\nu2\|d_i\|^2\nonumber\\
    &=\frac{\|s_i\|^2}{2\nu}
    -\frac\nu2\left(\|d_i\|-\frac{\|s_i\|}{\nu}\right)^2
    \leq\frac{\|s_i\|^2}{2\nu}.\label{eq:app-regret-square}
  \end{align}
  Taking the supremum over $z_i\geq0$ proves the regret bound.
  \label{app:independent-regret}
  The independent payoff expansion has the same form, giving the identical
  bound for each independent firm's objective.
\end{proof}

\subsubsection{Accuracy after capacity normalization}\label{app:observation}

Because composition divides capacities by total mass, a small capacity error
can matter when that mass is small.
For feasible profiles, the product norm bounds each block,
\begin{equation}\label{eq:app-block-bound}
  \|\widehat q_i-q_i^*\|^2\leq\sum_j\|\widehat q_j-q_j^*\|^2
  =\|\widehat q-q^*\|^2,
\end{equation}
so the joint certificate supplies a conservative error radius for every firm.

\begin{proposition}[Observation error under
  normalization]\label{prop:normalization}
  Suppose $\|\widehat q_i-q_i^*\|\leq e$ and
  $\widehat S_i=\mathbf1^\top\widehat q_i>\sqrt m\,e$.
  Then $S_i^*>0$ and
  \begin{equation}\label{eq:normalization-bound}
    \left\|\frac{\widehat q_i}{\widehat S_i}-\frac{q_i^*}{S_i^*}\right\|
    \leq\frac{2\sqrt m\,e}{\widehat S_i-\sqrt m\,e}.
  \end{equation}
\end{proposition}

\begin{proof}
  Cauchy--Schwarz bounds the mass error by
  \begin{equation}\label{eq:app-mass-error}
    |\widehat S_i-S_i^*|
    =|\langle\mathbf1,\widehat q_i-q_i^*\rangle|\leq\sqrt m\,e.
  \end{equation}
  Hence $S_i^*\geq\widehat S_i-\sqrt m\,e>0$.
  Nonnegativity gives $\|\widehat q_i\|\leq\widehat S_i$, and the identity
  \begin{equation}\label{eq:app-normalization-split}
    \frac{\widehat q_i}{\widehat S_i}-\frac{q_i^*}{S_i^*}
    =\frac{\widehat q_i-q_i^*}{S_i^*}
    +\frac{S_i^*-\widehat S_i}{\widehat S_i S_i^*}\widehat q_i
  \end{equation}
  therefore implies
  \begin{equation}\label{eq:app-normalization-chain}
    \left\|\frac{\widehat q_i}{\widehat S_i}-\frac{q_i^*}{S_i^*}\right\|
    \leq\frac{(1+\sqrt m)e}{S_i^*}
    \leq\frac{2\sqrt m\,e}{\widehat S_i-\sqrt m\,e},
  \end{equation}
  where the last step uses $m\geq1$.
\end{proof}

The benchmark experiments accept profiles using the independently evaluated
relative stationarity-distance target \ppnume[0]{ir-solve-target}.
If multiple starts give error radii $e_a,e_b$, uniqueness and the triangle
inequality bound their distance by $e_a+e_b$; agreement across starts
supplements the residual check.
These inequalities hold in real arithmetic, while their numerical evaluations
use floating-point capacities and residuals without interval-arithmetic
enclosures.

\section{Recovery diagnostics}\label{app:recovery-diagnostics}

\subsection{Directional and magnitude assessment}\label{app:directional-diagnostics}

The assessment in Section~\ref{sec:simulation} uses
\ppnum[0]{ir-validation-redraws} article redraws per history, with
stationary-law, reporting-only and bounded shock controls alongside the
common-reporting coverage ladder.
All sixteen mechanism histories satisfy the positivity, baseline-price,
solver and active-domain requirements.
Each history has a common-reporting operating-law target $\beta_h^0$ and a
reported-law target $\beta_h^\lambda$; sampled coefficients can be
compared with either target according to the measurement question.
Sign frequencies and coefficient errors summarize these comparisons.

For an event indicator $S_{hr}$ in history $h$ and article redraw $r$, let
$S_{hr}=0$ when the estimate is undefined.
With $H=16$ histories and $R$ redraws per history, write
$\widehat p_h=R^{-1}\sum_r S_{hr}$.
The equally weighted frequency and stratified Monte Carlo standard error are
\begin{equation}\label{eq:sign-mcse}
  \widehat p=\frac1H\sum_h\widehat p_h,\qquad
  \widehat{\operatorname{se}}(\widehat p)=
  \left[\frac1{H^2(R-1)}\sum_h\widehat p_h(1-\widehat p_h)\right]^{1/2}.
\end{equation}
The histories are enumerated economic paths, so this standard error measures
article-redraw uncertainty within those paths.

Numerical sign classification uses an arithmetic error allowance fixed before
the control results were inspected.
With $\eta=10^{-12}$, axis-transport allowances propagate through squared
distance, approach and dated financial weights to the centered numerator.
A sign is resolved only when that numerator exceeds its propagated
allowance in absolute value; constructed population zeros retain their
exact-zero status.
This arithmetic convention is separate from the economic solver certificates
and preserves a zero financial denominator as an undefined coefficient.

The identity
\begin{equation}\label{eq:recovery-accounting}
  \widehat\beta^\lambda-\beta^0
  =(\widehat\beta^\lambda-\beta^\lambda)+(\beta^\lambda-\beta^0)
\end{equation}
separates article measurement error from observation distortion.
Only the first term vanishes with infinitely many articles under the
maintained sampling law.
For recovery of the operating-law target, normalized RMSE is
\begin{equation}\label{eq:recovery-rmse}
  \mathcal R_s^\lambda=
  \frac{\sqrt{16^{-1}\sum_h E_r[(\widehat\beta_{h,r,s}^\lambda-\beta_h^0)^2]}}
  {\sqrt{16^{-1}\sum_h(\beta_h^0)^2}}.
\end{equation}
Operational recovery requires this normalized RMSE plus twice its stratified
Monte Carlo standard error to be at most 25\%.
The 25\% threshold is a declared accuracy convention applied only to the
tested budgets.

Signed coefficient bias and its Monte Carlo standard error are in
coefficient units; normalized RMSE uses the oracle RMS as its scale.
For squared errors $L_{hr}$ and oracle RMS $B$, the RMSE standard error uses
\begin{equation}\label{eq:recovery-mcse}
  \widehat{\operatorname{se}}(\widehat{\mathcal R})=
  \frac{\sqrt{H^{-2}R^{-1}\sum_h\widehat{\operatorname{Var}}_r(L_{hr})}}
  {2B\sqrt{H^{-1}R^{-1}\sum_{h,r}L_{hr}}}.
\end{equation}
This estimates article-redraw Monte Carlo uncertainty, conditional on the
enumerated histories, rather than uncertainty about another economic
environment.
If an estimate is undefined, the summary records its count and the affected
recovery criterion cannot pass.
The common-reporting population oracle has zero recovery error by definition;
the changing-reporting population error is calculated exactly.

These summaries are equal-history performance averages; the theoretical expected
moment does not require each finite-history ratio to be positive.
The ratio of expected moments differs from the average of history-specific ratios.
Budget comparisons pair each history's nested redraws rather than treating
budgets or dyads as independent replications.
All attempted estimates remain in their eligible rate denominators;
execution failures, failed economic premises and undefined statistics are
counted separately.
For zero population coefficients, absolute bias, MAE and RMSE remain
available, while sign recovery and normalized RMSE are unavailable.

% AUTO-GENERATED from saved Paper 6 results; do not edit.
\begin{table}[htbp]
\centering
\caption{Directional measurement under fixed economic histories. Targets distinguish common-reporting operating coefficients from reported-law coefficients, including the zero operating target in reporting-only cells. Rates are percentages (stratified Monte Carlo SE in percentage points), not test power. RMSE is in coefficient units against each history's matching target. Appendix~\ref{app:directional-diagnostics} defines resolved nonzero oracles, eligible rate denominators and undefined-attempt accounting. The main table's mean target is an arithmetic mean of history coefficients, not a ratio of expected moments. Each row retains \ppnum[0]{ir-validation-histories} histories and \ppnum[0]{ir-validation-redraws} redraws per history. All attempts enter the rate denominators. No undefined estimates, failed premises or domain departures. Correct sign is unavailable for zero or unresolved targets.}
\label{tab:directional-validation-details}
\begin{tabular}{@{}l@{\enspace}l@{\enspace}r@{\enspace}r@{\enspace}r@{}}
\toprule
Cell & Target & Correct & Positive & RMSE \\
\midrule
\ppvalue{ir-validation-0-label} & \ppvalue{ir-validation-0-target} & \ppnum[1]{ir-validation-0-correct-pct}(\ppnum[2]{ir-validation-0-correct-se-pct}) & \ppnum[1]{ir-validation-0-positive-pct}(\ppnum[2]{ir-validation-0-positive-se-pct}) & \ppnum[3]{ir-validation-0-rmse} \\
\ppvalue{ir-validation-1-label} & \ppvalue{ir-validation-1-target} & \ppnum[1]{ir-validation-1-correct-pct}(\ppnum[2]{ir-validation-1-correct-se-pct}) & \ppnum[1]{ir-validation-1-positive-pct}(\ppnum[2]{ir-validation-1-positive-se-pct}) & \ppnum[3]{ir-validation-1-rmse} \\
\ppvalue{ir-validation-2-label} & \ppvalue{ir-validation-2-target} & \ppnum[1]{ir-validation-2-correct-pct}(\ppnum[2]{ir-validation-2-correct-se-pct}) & \ppnum[1]{ir-validation-2-positive-pct}(\ppnum[2]{ir-validation-2-positive-se-pct}) & \ppnum[3]{ir-validation-2-rmse} \\
\ppvalue{ir-validation-3-label} & \ppvalue{ir-validation-3-target} & \ppnum[1]{ir-validation-3-correct-pct}(\ppnum[2]{ir-validation-3-correct-se-pct}) & \ppnum[1]{ir-validation-3-positive-pct}(\ppnum[2]{ir-validation-3-positive-se-pct}) & \ppnum[3]{ir-validation-3-rmse} \\
\ppvalue{ir-validation-4-label} & \ppvalue{ir-validation-4-target} & \ppnum[1]{ir-validation-4-correct-pct}(\ppnum[2]{ir-validation-4-correct-se-pct}) & \ppnum[1]{ir-validation-4-positive-pct}(\ppnum[2]{ir-validation-4-positive-se-pct}) & \ppnum[3]{ir-validation-4-rmse} \\
\ppvalue{ir-validation-5-label} & \ppvalue{ir-validation-5-target} & \ppnum[1]{ir-validation-5-correct-pct}(\ppnum[2]{ir-validation-5-correct-se-pct}) & \ppnum[1]{ir-validation-5-positive-pct}(\ppnum[2]{ir-validation-5-positive-se-pct}) & \ppnum[3]{ir-validation-5-rmse} \\
\ppvalue{ir-validation-6-label} & \ppvalue{ir-validation-6-target} & \ppmissing & \ppnum[1]{ir-validation-6-positive-pct}(\ppnum[2]{ir-validation-6-positive-se-pct}) & \ppnum[3]{ir-validation-6-rmse} \\
\ppvalue{ir-validation-7-label} & \ppvalue{ir-validation-7-target} & \ppmissing & \ppnum[1]{ir-validation-7-positive-pct}(\ppnum[2]{ir-validation-7-positive-se-pct}) & \ppnum[3]{ir-validation-7-rmse} \\
\ppvalue{ir-validation-8-label} & \ppvalue{ir-validation-8-target} & \ppmissing & \ppnum[1]{ir-validation-8-positive-pct}(\ppnum[2]{ir-validation-8-positive-se-pct}) & \ppnum[3]{ir-validation-8-rmse} \\
\ppvalue{ir-validation-9-label} & \ppvalue{ir-validation-9-target} & \ppnum[1]{ir-validation-9-correct-pct}(\ppnum[2]{ir-validation-9-correct-se-pct}) & \ppnum[1]{ir-validation-9-positive-pct}(\ppnum[2]{ir-validation-9-positive-se-pct}) & \ppnum[3]{ir-validation-9-rmse} \\
\ppvalue{ir-validation-10-label} & \ppvalue{ir-validation-10-target} & \ppmissing & \ppnum[1]{ir-validation-10-positive-pct}(\ppnum[2]{ir-validation-10-positive-se-pct}) & \ppnum[3]{ir-validation-10-rmse} \\
\ppvalue{ir-validation-11-label} & \ppvalue{ir-validation-11-target} & \ppnum[1]{ir-validation-11-correct-pct}(\ppnum[2]{ir-validation-11-correct-se-pct}) & \ppnum[1]{ir-validation-11-positive-pct}(\ppnum[2]{ir-validation-11-positive-se-pct}) & \ppnum[3]{ir-validation-11-rmse} \\
\ppvalue{ir-validation-12-label} & \ppvalue{ir-validation-12-target} & \ppnum[1]{ir-validation-12-correct-pct}(\ppnum[2]{ir-validation-12-correct-se-pct}) & \ppnum[1]{ir-validation-12-positive-pct}(\ppnum[2]{ir-validation-12-positive-se-pct}) & \ppnum[3]{ir-validation-12-rmse} \\
\bottomrule
\end{tabular}
\end{table}

To examine the local approximation, we halve and double the baseline
opportunity innovation at observed article coverage.
All sixteen histories remain in the economic domain and retain positive
tangent and finite-capacity oracles.
Correct-sign recovery is
\ppnum[1]{ir-validation-half-correct-pct}\% and
\ppnum[1]{ir-validation-double-correct-pct}\%, respectively.
The financial denominator scales approximately with the squared shock
amplitude while the population ratio changes little, so this variation
chiefly changes signal relative to article noise.

\subsection{Historical recovery and multi-activity
comparisons}\label{app:historical-recovery}

The 128-redraw study provides the common- and changing-reporting comparison
and the multi-activity dimension checks alongside the current controls.
Its retained magnitude results also document the 25\% recovery benchmark.
Its positive-estimate diagnostic uses
$S_{hr}=\mathbf1\{\widehat\beta_{h,r,s}^\lambda>0\}$ for defined
estimates and zero otherwise, with $R=128$ in \eqref{eq:sign-mcse}.
It precedes the resolved-sign convention used by the current assessment.
Table~\ref{tab:directional-recovery} reports these historical frequencies.

The historical 128-redraw study checks positive opportunities, baseline prices
and active capacities, the solver-distance bound, and agreement of the
operating and affine active domains for every canonical history.
The largest relative solver bound is
\ppnume[2]{ir-maximum-solver-bound}; the minimum canonical capacity and
baseline price are \ppnum[6]{ir-minimum-capacity} and
\ppnum[6]{ir-minimum-baseline-price}.
These checks yield \ppnum[0]{ir-failures} failed histories and
\ppnum[0]{ir-active-domain-violations} active-domain violations.

% AUTO-GENERATED from saved Paper 6 results; do not edit.
\begin{table}[htbp]
\centering
\caption{Historical positive-estimate frequencies from finite article panels, using the original literal $\widehat\beta>0$ convention. All sixteen common-reporting population coefficients are positive. Frequencies average 128 redraws within each history and weight histories equally. SE is the stratified Monte Carlo standard error in percentage points. Zero and undefined estimates do not count as positive; undefined draws remain in the denominator.}
\label{tab:directional-recovery}
\begin{tabular}{llrrr}
\toprule
Reporting & \shortstack{Article-count\\multiplier} & Positive (\%) & SE (pp) & Undefined \\
\midrule
Common & 1 & \ppnum[1]{ir-common-1-positive-pct} & \ppnum[2]{ir-common-1-positive-se-pct} & \ppnum[0]{ir-common-1-sign-undefined} \\
Common & 4 & \ppnum[1]{ir-common-4-positive-pct} & \ppnum[2]{ir-common-4-positive-se-pct} & \ppnum[0]{ir-common-4-sign-undefined} \\
Common & 16 & \ppnum[1]{ir-common-16-positive-pct} & \ppnum[2]{ir-common-16-positive-se-pct} & \ppnum[0]{ir-common-16-sign-undefined} \\
Changing & 1 & \ppnum[1]{ir-changing-1-positive-pct} & \ppnum[2]{ir-changing-1-positive-se-pct} & \ppnum[0]{ir-changing-1-sign-undefined} \\
Changing & 4 & \ppnum[1]{ir-changing-4-positive-pct} & \ppnum[2]{ir-changing-4-positive-se-pct} & \ppnum[0]{ir-changing-4-sign-undefined} \\
Changing & 16 & \ppnum[1]{ir-changing-16-positive-pct} & \ppnum[2]{ir-changing-16-positive-se-pct} & \ppnum[0]{ir-changing-16-sign-undefined} \\
\bottomrule
\end{tabular}
\end{table}

\begin{figure}[tb]
  \centering
  \import{images/}{information_recovery\ppfiguresuffix.pgf}
  \caption{Article information and reporting distortion in the fixed recovery
    experiment.
    Left: normalized recovery RMSE at observed article coverage and count
    multipliers four and sixteen, with bars of two stratified Monte Carlo
    standard errors; population-law points show the exact
    infinite-article error.
    The horizontal line is the 25\% operational convention; the left ordinate
    is logarithmic above that convention and linear below it.
    Right: exact population distortion across the specified log-intensity
    amplitudes, including a fixed-capacity reporting-only observation
    counterfactual.
    Every calculation retains all sixteen canonical economic histories.
    Financial news is calculated from the local valuation response to permanent
    opportunity innovations along the same operating path.
  Reporting amplitudes are illustrative benchmarks rather than empirical calibrations.}
  \label{fig:information-recovery}
\end{figure}

Table~\ref{tab:recovery-diagnostics} retains the historical signed bias,
normalized RMSE and their Monte Carlo standard errors.

% AUTO-GENERATED from saved Paper 6 results; do not edit.
\begin{table}[htbp]
\centering
\caption{Recovery diagnostics against the finite-capacity common-reporting oracle. Bias and its Monte Carlo standard error (SE) are in coefficient units; RMSE and its SE are normalized by oracle RMS. Undefined estimates remain counted. Histories receive equal weight.}
\label{tab:recovery-diagnostics}
\begin{tabular}{llrrrrr}
\toprule
Reporting & \shortstack{Article-count\\multiplier} & Bias & SE & RMSE & SE & Undefined \\
\midrule
Common & 1 & \ppnum[5]{ir-common-1-bias} & \ppnum[5]{ir-common-1-bias-se} & \ppnum[3]{ir-common-1-rmse} & \ppnum[3]{ir-common-1-rmse-se} & \ppnum[0]{ir-common-1-undefined} \\
Common & 4 & \ppnum[5]{ir-common-4-bias} & \ppnum[5]{ir-common-4-bias-se} & \ppnum[3]{ir-common-4-rmse} & \ppnum[3]{ir-common-4-rmse-se} & \ppnum[0]{ir-common-4-undefined} \\
Common & 16 & \ppnum[5]{ir-common-16-bias} & \ppnum[5]{ir-common-16-bias-se} & \ppnum[3]{ir-common-16-rmse} & \ppnum[3]{ir-common-16-rmse-se} & \ppnum[0]{ir-common-16-undefined} \\
Changing & 1 & \ppnum[5]{ir-changing-1-bias} & \ppnum[5]{ir-changing-1-bias-se} & \ppnum[3]{ir-changing-1-rmse} & \ppnum[3]{ir-changing-1-rmse-se} & \ppnum[0]{ir-changing-1-undefined} \\
Changing & 4 & \ppnum[5]{ir-changing-4-bias} & \ppnum[5]{ir-changing-4-bias-se} & \ppnum[3]{ir-changing-4-rmse} & \ppnum[3]{ir-changing-4-rmse-se} & \ppnum[0]{ir-changing-4-undefined} \\
Changing & 16 & \ppnum[5]{ir-changing-16-bias} & \ppnum[5]{ir-changing-16-bias-se} & \ppnum[3]{ir-changing-16-rmse} & \ppnum[3]{ir-changing-16-rmse-se} & \ppnum[0]{ir-changing-16-undefined} \\
\bottomrule
\end{tabular}
\end{table}

The appendix check varies geometric dimension over 2, 4, 8 and 16 while
holding 100 firms, eight activities, 64 demand nodes and $\kappa=1$ fixed.
It uses common reporting and the observed and sixteenfold article budgets,
with the same dated news histories and monthly discount convention.
These are generic multi-activity economies; the two-opportunity sign
mechanism is not imposed.
Table~\ref{tab:recovery-dimensions} reports their recovery and economic-domain checks.

% AUTO-GENERATED from saved Paper 6 results; do not edit.
\begin{table}[htbp]
\centering
\caption{Bounded dimension check under common reporting. All cells hold 100 firms, eight activities, 64 demand nodes and adjustment weight one fixed. RMSE is normalized by oracle RMS, with stratified Monte Carlo SE. All 8 configurations completed with no undefined estimates.}
\label{tab:recovery-dimensions}
\begin{tabular}{llrr}
\toprule
Dimension & \shortstack{Article-count\\multiplier} & RMSE & SE \\
\midrule
2 & 1 & \ppnum[3]{ir-dimension-2-1-rmse} & \ppnum[3]{ir-dimension-2-1-se} \\
2 & 16 & \ppnum[3]{ir-dimension-2-16-rmse} & \ppnum[3]{ir-dimension-2-16-se} \\
4 & 1 & \ppnum[3]{ir-dimension-4-1-rmse} & \ppnum[3]{ir-dimension-4-1-se} \\
4 & 16 & \ppnum[3]{ir-dimension-4-16-rmse} & \ppnum[3]{ir-dimension-4-16-se} \\
8 & 1 & \ppnum[3]{ir-dimension-8-1-rmse} & \ppnum[3]{ir-dimension-8-1-se} \\
8 & 16 & \ppnum[3]{ir-dimension-8-16-rmse} & \ppnum[3]{ir-dimension-8-16-se} \\
16 & 1 & \ppnum[3]{ir-dimension-16-1-rmse} & \ppnum[3]{ir-dimension-16-1-se} \\
16 & 16 & \ppnum[3]{ir-dimension-16-16-rmse} & \ppnum[3]{ir-dimension-16-16-se} \\
\bottomrule
\end{tabular}
\end{table}

\section{Moment representation and uncertainty
assessment}\label{app:annual-gmm}

\subsection{Independent firm-history sampling}\label{app:history-sampling}

The same firm appears as a target and as a reference across many comparisons.
To describe how this reuse enters the moments, write them in terms of complete
firm triples.

For an unordered triple $T=\{i,j,k\}$, define symmetric kernels by averaging
its three target assignments:
\begin{equation}\label{eq:complete-triples}
  H^q_y(T)=\frac13\sum_{i\in T}(R_{ijy}-R_{iky})(A_{ijy}-A_{iky}),\qquad
  H^d_y(T)=\frac13\sum_{i\in T}(R_{ijy}-R_{iky})^2,
\end{equation}
where $j,k$ are the other two members; their ordering does not affect
the product.
Expanding each centered covariance into reference-pair contrasts gives
\begin{equation}\label{eq:triple-identity}
  (q_y,d_y)=c_N\binom N3^{-1}\sum_T(H^q_y(T),H^d_y(T)),\qquad
  c_N=\frac{N-2}{2(N-1)}.
\end{equation}
Identity~\eqref{eq:triple-identity} expresses the all-reference moments through
overlapping triples, making their shared firm histories explicit rather than
treating the triples as independent observations.

Under independent sampling of complete firm histories, the contribution of
one firm must include every triple in which it appears.
A conditional H\'ajek benchmark projects the order-three kernels onto those
histories, preserving covariance across all years.
For $U_y=\binom N3^{-1}\sum_T H_y(T)$, a finite incident-triple contribution is
$\widehat\phi_{iy}=3c_N[\binom{N-1}{2}^{-1}\sum_{T\ni i}H_y(T)-U_y]$.
Applying the ratio derivative gives
\begin{equation}\label{eq:conditional-sandwich}
  \widehat\psi_i=
  \frac{\sum_y(\widehat\phi^q_{iy}-\widehat\beta\widehat\phi^d_{iy})}{\sum_y
  d_y},\qquad
  \widehat V_{\beta}=\frac{\widehat{\operatorname{Var}}_i(\widehat\psi_i)}{N}.
\end{equation}
Interpreting this benchmark as population uncertainty requires conditionally
iid complete histories given a fixed common environment, fixed symmetric
square-integrable kernels, a positive limiting denominator, a nondegenerate
first projection and a consistent covariance estimator; the empirical
incident projection is not automatically an unbiased population influence
estimate.

To assess finite-sample calibration under the iid-history assumption, let
each firm draw independent $X_i,E_i\sim N(0,1)$ and an independent
Rademacher $U_i$, retaining these variables across all years and roles.
For $t=0,\ldots,3$, use the generated arrays
\begin{equation}\label{eq:inference-fixture}
  R_{ijt}=c_t X_j,\qquad
  A_{ijt}=\beta R_{ijt}+g_t(E_j+\tfrac{1}{2}E_i X_j)+tU_i,
\end{equation}
where $c=(1,2,1,2)$ and $g=(1,1,2,2)$, with self comparisons excluded.
Reference centering removes $tU_i$.
The population ratio of expected pooled moments is $\beta$, its first
influence is $0.9E_i(X_i+1/2)$, and its asymptotic variance is $1.0125/N$.
This simulation design supplies a known sampling target while preserving
cross-year and cross-role dependence.
Its firms are sampled independently rather than jointly solving the
competitive economy.

For $H_0:\beta\leq0$ against $H_1:\beta>0$, the candidate lower bound is
\begin{equation}\label{eq:one-sided-candidate}
  L=\widehat\beta-z_{0.95}\sqrt{\widehat V_\beta},\qquad
  z_{0.95}=\ppnum[6]{ir-inference-z},
\end{equation}
with rejection when $L>0$.
We compare 16-firm and 100-firm panels, with
\ppnum[0]{ir-inference-n100-replications} independent replications per
target at each size and \ppnum[0]{ir-inference-n100-years} years.
Within each size, the zero and
$\beta=\ppnum[1]{ir-inference-n100-alternative-target}$ targets use the
same firm draws.
The fixed replication count bounds the worst-case Bernoulli Monte Carlo
standard error at one percentage point.
Table~\ref{tab:iid-inference-audit} reports target containment, null and
alternative rejection, and the ratio of mean candidate standard error to
empirical dispersion.
Increasing the sample size brings coverage closer to the nominal 95\%
within this sampling design, although some undercoverage remains.

% AUTO-GENERATED from saved Paper 6 results; do not edit.
\begin{table}[htbp]
\centering
\caption{One-sided lower-bound performance under independent firm-history sampling. Nominal 95\% lower bounds use the one-sided normal critical value. Coverage and rejection of $H_0:\beta\leq0$ are percentages (Monte Carlo SE in percentage points). SE/SD compares mean candidate standard error with empirical dispersion; unavailable bounds remain in the performance denominators. Both panel sizes retain four dependent years and each firm's target and reference roles. Each row uses \ppnum[0]{ir-inference-audit-n16-null-attempts} attempted replications, with \ppnum[0]{ir-inference-audit-n16-null-available} usable bounds.}
\label{tab:iid-inference-audit}
\begin{tabular}{@{}r@{\enspace}r@{\enspace}r@{\enspace}r@{\enspace}r@{}}
\toprule
Firms & Target & Coverage & Reject $H_0$ & SE/SD \\
\midrule
\ppnum[0]{ir-inference-audit-n16-firms} & \ppnum[1]{ir-inference-audit-n16-null-target} & \ppnum[1]{ir-inference-audit-n16-null-coverage-pct}(\ppnum[2]{ir-inference-audit-n16-null-coverage-se-pct}) & \ppnum[1]{ir-inference-audit-n16-null-rejection-pct}(\ppnum[2]{ir-inference-audit-n16-null-rejection-se-pct}) & \ppnum[3]{ir-inference-audit-n16-null-se-ratio} \\
\ppnum[0]{ir-inference-audit-n16-firms} & \ppnum[1]{ir-inference-audit-n16-alternative-target} & \ppnum[1]{ir-inference-audit-n16-alternative-coverage-pct}(\ppnum[2]{ir-inference-audit-n16-alternative-coverage-se-pct}) & \ppnum[1]{ir-inference-audit-n16-alternative-rejection-pct}(\ppnum[2]{ir-inference-audit-n16-alternative-rejection-se-pct}) & \ppnum[3]{ir-inference-audit-n16-alternative-se-ratio} \\
\ppnum[0]{ir-inference-n100-firms} & \ppnum[1]{ir-inference-n100-null-target} & \ppnum[1]{ir-inference-n100-null-coverage-pct}(\ppnum[2]{ir-inference-n100-null-coverage-se-pct}) & \ppnum[1]{ir-inference-n100-null-rejection-pct}(\ppnum[2]{ir-inference-n100-null-rejection-se-pct}) & \ppnum[3]{ir-inference-n100-null-se-ratio} \\
\ppnum[0]{ir-inference-n100-firms} & \ppnum[1]{ir-inference-n100-alternative-target} & \ppnum[1]{ir-inference-n100-alternative-coverage-pct}(\ppnum[2]{ir-inference-n100-alternative-coverage-se-pct}) & \ppnum[1]{ir-inference-n100-alternative-rejection-pct}(\ppnum[2]{ir-inference-n100-alternative-rejection-se-pct}) & \ppnum[3]{ir-inference-n100-alternative-se-ratio} \\
\bottomrule
\end{tabular}
\end{table}

\subsection{Article-selection uncertainty}\label{app:article-selection}

For article selection, independent generated archives supply the repeated
samples while firms, economic histories and reported population laws remain
fixed.
We retain all sixteen histories and the common-reporting mechanism,
stationary-law control and reporting-only law at amplitude 0.005.
Each history has \ppnum[0]{ir-redraw-outer-replications} outer archives
at observed expanded coverage.
Within each firm-year, the canonical and additional article strata retain
their frozen sample sizes and monthly date probabilities.
An outer archive draws dates within each stratum and article atoms conditional
on those dates from the reported law.

Each of its \ppnum[0]{ir-redraw-inner-redraws} inner redraws resamples the
empirical joint date--atom counts separately by stratum at unchanged sizes.
The same counts determine article geometry and article-timing weights, and
each origin draw is reused in every target and reference role.

To assess whether the redraw ranges capture the reported-law coefficient,
we use the population projection for the prescribed stratum-specific date
distributions in each fixed history as the primary target.
The archive projection and the population projection conditional on its
article dates are also recorded.
Linear empirical quantiles give the 2.5th--97.5th percentile range and fifth
percentile.
We measure how often the range contains the target, how often the target
lies above the fifth percentile, and how often that percentile exceeds zero.
Table~\ref{tab:article-redraw-validation} averages these events equally
across histories; its Monte Carlo errors use independent outer archives,
not the inner redraws as separate replications.
Quantiles use defined inner estimates, matching the corporate procedure;
all attempted inner draws and unavailable outer outputs remain counted.
The stationary law has a zero target, while the reporting-only law has a
potentially nonzero reported-law target despite fixed capacities.

% AUTO-GENERATED from saved Paper 6 results; do not edit.
\begin{table}[htbp]
\centering
\caption{Conditional article-redraw performance under simulated reporting laws at observed expanded support. Central coverage uses the 2.5th--97.5th percentile interval; q.05 coverage uses the lower fifth percentile. Each cell gives the equal-history percentage (stratified Monte Carlo SE in percentage points). Coverage concerns the matching history-specific reported-law target. Positive q.05 is a zero-null rejection frequency only for the stationary control; the other rows have positive reported-law targets. Each percentage uses \ppnum[0]{ir-redraw-mechanism-central-coverage-attempts} attempted outer panels, with \ppnum[0]{ir-redraw-mechanism-central-coverage-available} available bounds. The design retains \ppnum[0]{ir-redraw-histories} histories, \ppnum[0]{ir-redraw-outer-replications} outer panels per history and \ppnum[0]{ir-redraw-inner-redraws} inner redraws per panel.}
\label{tab:article-redraw-validation}
\begin{tabular}{@{}l@{\enspace}r@{\enspace}r@{\enspace}r@{}}
\toprule
Design & Central coverage & q.05 coverage & Positive q.05 \\
\midrule
Common reporting & \ppnum[1]{ir-redraw-mechanism-central-coverage-pct}(\ppnum[2]{ir-redraw-mechanism-central-coverage-se-pct}) & \ppnum[1]{ir-redraw-mechanism-q05-coverage-pct}(\ppnum[2]{ir-redraw-mechanism-q05-coverage-se-pct}) & \ppnum[1]{ir-redraw-mechanism-positive-q05-pct}(\ppnum[2]{ir-redraw-mechanism-positive-q05-se-pct}) \\
Stationary reporting & \ppnum[1]{ir-redraw-stationary-central-coverage-pct}(\ppnum[2]{ir-redraw-stationary-central-coverage-se-pct}) & \ppnum[1]{ir-redraw-stationary-q05-coverage-pct}(\ppnum[2]{ir-redraw-stationary-q05-coverage-se-pct}) & \ppnum[1]{ir-redraw-stationary-positive-q05-pct}(\ppnum[2]{ir-redraw-stationary-positive-q05-se-pct}) \\
Reporting only & \ppnum[1]{ir-redraw-reporting-central-coverage-pct}(\ppnum[2]{ir-redraw-reporting-central-coverage-se-pct}) & \ppnum[1]{ir-redraw-reporting-q05-coverage-pct}(\ppnum[2]{ir-redraw-reporting-q05-coverage-se-pct}) & \ppnum[1]{ir-redraw-reporting-positive-q05-pct}(\ppnum[2]{ir-redraw-reporting-positive-q05-se-pct}) \\
\bottomrule
\end{tabular}
\end{table}

An exact empirical transport problem can have a nonunique optimal dual face,
so fixed-support mass perturbations may have different left and right
derivatives.
Finite directional sensitivities and the linear ratio derivative alone do
not establish smooth article influence functions; finite-sample containment
under the specified observation law is distinct from a general bootstrap
validity result.

The iid-history and nested-archive experiments concern distinct sampling
populations.
Jointly solved competitors do not become independent histories by conditioning
on common news, and the corporate archive of Nasdaq financial news, collected
in December
2022--January 2023 with operating reports and market commentary from multiple
publishers, has no established iid-history sampling design.
Its article redraws condition on the retained firms, returns and years.
The application therefore reports article-selection sensitivity rather than
a corporate population confidence bound or $p$-value.

\subsection{Sensitivity of the pooled projection}\label{app:projection-sensitivity}

The complete-triple representation gives all-role deletion moments by
removing every triple incident to the deleted firm and using the remaining
sample's normalization.
Consequently a firm is removed everywhere it appears, rather than only as a
target or as one financial reference.

When annual ratios differ, changing their weights changes the pooled
summary.
A fixed-weight annual GMM criterion makes that dependence explicit alongside
the moment ratio in \eqref{eq:directional-ratio}.
For $q=(q_y)_y$, $d=(d_y)_y$ and a fixed symmetric positive-definite weight
$W$, minimizing $(q-\beta d)^\top W(q-\beta d)$ gives
\begin{equation}\label{eq:annual-gmm}
  \beta_W=\frac{d^\top Wq}{d^\top Wd},\qquad
  d\beta_W=\frac{d^\top W\,dq+(q-2\beta_W d)^\top W\,dd}{d^\top Wd}.
\end{equation}
Both annual numerators and denominators vary; the term involving $dd$
retains the common-slope misspecification residual.
At $W=I$, annual ratios receive weights proportional to $d_y^2$.
The original pool weights them by $d_y$.
For the primary cell, the fixed-identity slope is \ppnum[6]{ir-observation-gmm-beta} and
2020 receives \ppnum[1]{ir-observation-gmm-2020-weight-pct}\% effective weight rather
than \ppnum[1]{ir-primary-2020-weight-pct}\%.
This is a different annual compromise, not an efficiency improvement to an
exactly identified ratio.
No estimated optimal weighting or chi-square reference for a
specification test is used.

The deletion checks remove observed firms, but measurement errors can also
alter the pooled alignment.
A deterministic error frontier bounds such changes under hypothetical
error limits.
Write observed centered arrays as $R=R^{\mathrm{core}}+v$ and
$A=A^{\mathrm{core}}+e$.
Then the core alignment is exactly
\begin{equation}\label{eq:error-frontier}
  \langle R-v,A-e\rangle
  =\langle R,A\rangle-\langle R,e\rangle-\langle v,A\rangle+\langle v,e\rangle.
\end{equation}
If the unobserved errors satisfy $\|v\|\leq\lambda_R\|R\|$ and
$\|e\|\leq\lambda_A\|A\|$, Cauchy--Schwarz bounds the absolute departure by
$\|R\|\|A\|(\lambda_R+\lambda_A+\lambda_R\lambda_A)$.
The primary normalized alignment is \ppnum[5]{ir-normalized-alignment}; equal
error budgets below \ppnum[2]{ir-observation-equal-budget-pct}\% preserve a positive
finite core alignment under this conservative bound.
Those percentages are hypothetical relative RMS error limits, not estimated
reporting errors or confidence levels.
Crossing the frontier removes the sufficient guarantee without constructing a
sign reversal.

\section{Observation-profile audit}\label{app:observation-audit}

The restricted movement profile reduces held-out squared reconstruction error
by only \ppnum[2]{ir-observation-primary-heldout-pct}\% in the primary
specification, while two alternatives perform worse than zero centered change.
This audit assesses the shape of annual news movements separately from their
association with financial returns.

% AUTO-GENERATED from saved Paper 6 results; do not edit.
\begin{table}[htbp]
\centering
\caption{Historical held-out observation-profile audit. Skill is the percentage reduction in centered profile squared error relative to predicting zero change; negative values are worse than zero. The article-fold audit uses geometry alone and is a secondary observation-model diagnostic. Its held-out score does not validate the financial coefficient.}
\label{tab:observation-fit}
\begin{tabular}{lr}
\toprule
Cell & Skill (\%) \\
\midrule
Native, expanded & \ppnum[2]{ir-observation-d2560-expanded-heldout-pct} \\
Native, matched & \ppnum[2]{ir-observation-d2560-matched-heldout-pct} \\
Leading 64, expanded & \ppnum[2]{ir-observation-d64-expanded-heldout-pct} \\
Leading 64, matched & \ppnum[2]{ir-observation-d64-matched-heldout-pct} \\
\bottomrule
\end{tabular}
\end{table}

With profile and reference ordering fixed from 2018 article geometry,
the observation audit fits an annual scalar magnitude on one target-article
fold and evaluates it on the other, then reverses the roles.
Every firm remains in both folds; the fit uses article geometry alone.
Annual held-out scores are \ppnum[4]{ir-observation-heldout-2019},
\ppnum[4]{ir-observation-heldout-2020}, \ppnum[4]{ir-observation-heldout-2021} and
\ppnum[4]{ir-observation-heldout-2022}.
Negative scores in 2020 and 2021 show poor reconstruction within the same
firms and years.

The affine-rank calibration places \ppnum[0]{ir-observation-locality-count} of
\ppnum[0]{ir-observation-locality-defined} paths
(\ppnum[1]{ir-observation-locality-pct}\%)
at or beyond the distance to the nearest distinct reference distribution at
the initial date.
These geometric departures do not identify economic crossings: operating
shares are unidentified and the profile fit is imperfect.

\end{document}